\documentclass[fleqn,usenatbib]{mnras}

\usepackage{newtxtext,newtxmath}

\usepackage[T1]{fontenc}

\DeclareRobustCommand{\VAN}[3]{#2}
\let\VANthebibliography\thebibliography
\def\thebibliography{\DeclareRobustCommand{\VAN}[3]{##3}\VANthebibliography}

\usepackage{graphicx}	
\usepackage{amsmath}	
\usepackage{mathtools}
\usepackage{mathrsfs}
\usepackage{ulem}

\newcommand\vphiqm{\overline{v_{\varphi}^2}}

\newcommand\Deltac{\Delta_R}

\newcommand{\phis}{\Phi_*}
\newcommand{\phih}{\Phi_{\mathrm{h}}}
\newcommand{\rhoh}{\rho_{\mathrm{h}}}
\newcommand{\rhos}{\rho_*}

\newcommand{\bh}{b_{\mathrm{h}}}
\newcommand{\bM}{b_{\mathrm{M}}}

\newcommand{\vh}{v_{\mathrm{h}}}
\newcommand{\bplus}{\mathscr{B}^{+}}
\newcommand{\bminus}{\mathscr{B}^{-}}
\newcommand{\cplus}{\mathscr{C}^+}
\newcommand{\cminus}{\mathscr{C}^-}
\newcommand{\dplus}{\mathscr{D}^{+}}
\newcommand{\dminus}{\mathscr{D}^{-}}
\newcommand{\bpm}{\mathscr{B}^{\pm}}
\newcommand{\cpm}{\mathscr{C}^{\pm}}
\newcommand{\dpm}{\mathscr{D}^{\pm}}
\newcommand{\bzero}{\mathscr{B}^0}
\newcommand{\czero}{\mathscr{C}^0}
\newcommand{\dzero}{\mathscr{D}^0}

\newcommand{\opmzero}{\Omega^{\pm 0}}
\newcommand{\bpzero}{\mathscr{B}^{+0}}
\newcommand{\cpzero}{\mathscr{C}^{+0}}
\newcommand{\dpzero}{\mathscr{D}^{+0}}
\newcommand{\opzero}{\Omega^{+0}}

\newcommand{\dmzero}{\mathscr{D}^{-0}}

\newcommand{\bzplus}{\mathscr{B}_z^{+}}
\newcommand{\bzminus}{\mathscr{B}_z^{-}}
\newcommand{\czplus}{\mathscr{C}_z^+}
\newcommand{\czminus}{\mathscr{C}_z^-}

\newcommand{\dzminus}{\mathscr{D}_z^{-}}

\newcommand{\bzpm}{\mathscr{B}_z^{\pm}}
\newcommand{\czpm}{\mathscr{C}_z^{\pm}}
\newcommand{\dzpm}{\mathscr{D}_z^{\pm}}
\newcommand{\omegapm}{\Omega^{\pm}}
\newcommand{\omegazero}{\Omega^{0}}

\newcommand{\omegaminus}{\Omega^{-}}
\newcommand{\projz}{\mathrm{Pr}}
\newcommand{\vphim}{\overline{v_{\varphi}}}
\newcommand{\vrm}{\overline{v_{R}}}
\newcommand{\vzm}{\overline{v_{z}}}

\newtheorem{innercustomgeneric}{\customgenericname}
\providecommand{\customgenericname}{}
\newcommand{\newcustomtheorem}[2]{%
  \newenvironment{#1}[1]
  {%
   \renewcommand\customgenericname{#2}%
   \renewcommand\theinnercustomgeneric{##1}%
   \innercustomgeneric
  }
  {\endinnercustomgeneric}
}

\newcustomtheorem{customthm}{Theorem}

\newtheorem{theorem}{Theorem}[section]

\title[Anisotropy ansatz for axisymmetric systems]{Anisotropy ansatz for the axisymmetric Jeans equations}

\author[L. De Deo, L. Ciotti \& S. Pellegrini]{
Leonardo De Deo,$^{1,2,3}$\thanks{E-mail: leonardo.dedeo2@unibo.it}
Luca Ciotti,$^{1}$
and Silvia Pellegrini$^{1,2}$
\\
$^{1}$Department of Physics and Astronomy, University of Bologna, via Gobetti 93/2, 40129 Bologna, Italy\\
$^{2}$Istituto Nazionale di Astrofisica (INAF), Osservatorio di Astrofisica e Scienza dello Spazio di Bologna (OAS), Via Gobetti 93/3, Bologna 40129, Italy\\
$^{3}$ International PhD College - Collegio Superiore, University of Bologna, Italy
}

\date{Accepted 2024 April 8. Received 2024 March 27; in original form 2024 January 26}

\pubyear{2024}

\begin{document}
\label{firstpage}
\pagerange{\pageref{firstpage}--\pageref{lastpage}}
\maketitle

\begin{abstract}
The Jeans equations do not form a closed system, and to solve them a parametrization relating the velocity moments is often adopted. For axisymmetric models, a phenomenological choice (the \lq\lq $b$-ansatz") is widely used for the relation between the vertical ($\sigma_z^2$) and radial ($\sigma_R^2$) components of the velocity dispersion tensor, thus breaking their identity present in two-integral systems. However, the way in which the ansatz affects the resulting kinematical fields can be quite complicated, so that the analysis of these fields is usually performed only after numerically computing them. We present here a general procedure to study the properties of the ansatz-dependent fields $\vphiqm$, $\Delta=\vphiqm - \sigma_z^2$ and $\Deltac = \vphiqm - \sigma_R^2$. Specifically, the effects of the $b$-ansatz can be determined before solving the Jeans equations once the behaviour over the ($R,z$)-plane of three easy-to-build ansatz-independent functions is known. The procedure also constrains the ansatz to exclude unphysical results (as a negative $\vphiqm$). The method is illustrated by discussing the cases of three well-known galaxy models: the Miyamoto \& Nagai and Satoh disks, and the Binney logarithmic halo, for which the regions and the constraints on the ansatz values can be determined analytically; a two-component (Miyamoto \& Nagai plus logarithmic halo) model is also discussed.
\end{abstract}

\begin{keywords}
galaxies: kinematics and dynamics -- galaxies: structure -- galaxies: elliptical and lenticular, cD
\end{keywords}



\section{Introduction}

The Jeans equations (hereafter JEs) are one of the standard tools for the modelling of stellar systems
(e.g. \citealt{Binney2008gady.book.....B}, hereafter BT08; see also
\citealt{Ciotti2021isd..book.....C}, hereafter C21) and to extract
information from observations (e.g. \citealt{2015MNRAS.451.2723S},
\citealt{Cappellari2016ARA&A..54..597C},
\citealt{2016MNRAS.455.3680L},
\citealt{Zhu2023MNRAS.522.6326Z}). However, the JEs are moments of the more fundamental Boltzmann equation, and in the collisionless limit they suffer in general from a \lq\lq closure problem". Their closure can be obtained by assuming a (more or less motivated) dependence of the phase-space distribution function (hereafter DF) on the available integrals of motion. Alternatively, some relation between the velocity moments (usually between the velocity dispersion tensor components) can be imposed through some phenomenological \lq\lq ansatz". This latter approach is not as elegant and physically sound as deriving all the model properties from the DF, coupled with the Poisson equation and asking for self-consistency (e.g. \citealt{King1963AJ.....68Q.282K},
\citealt{1985MNRAS.217..735S}, \citealt{Bertin2008}), or using a DF built from actions (e.g. \citealt{Binney2010MNRAS.401.2318B}; see also \citealt{2015MNRAS.447.3060P}, \citealt{Binney2023MNRAS.520.1832B}), or reconstructing the DF
numerically with the Schwarzschild orbital superposition method
(\citealt{1979ApJ...232..236S}; see also, e.g., \citealt{1987ApJ...321..113S}, \citealt{Cappellari2007MNRAS.379..418C}, \citealt{2009MNRAS.393..641T}). However, the methods mentioned above are not always well suited for exploratory works, because in general the solution of the self-consistency problem requires non-trivial numerical work.

An important advantage of using the JEs is that the stellar density
distribution of the model can be chosen at the beginning to be in good
agreement with the observational data (for example by using specific
density profiles, or by multi-component modelling), and that the closure ansatz guarantees some direct control on the resulting
kinematics. Of course, the problem of the existence of a
non-negative DF for the obtained model remains in general
open.

In this work, we address two aspects worth a thorough investigation, both related to how a widely used closure ansatz (the $b$-ansatz, \citealt{Cappellari2008MNRAS.390...71C}) affects the solutions of the JEs: the first is to formalize a
general procedure, to be applied before solving the equations, to determine the constraints on this ansatz, as for example those coming from the request of positivity of $\vphiqm$, thus avoiding repeated and time-consuming numerical tests. The second is to
gain some qualitative understanding of its effects on the kinematical fields, before their construction. Moreover, as a widely used decomposition of $\vphiqm$ (e.g. that proposed by \citealt{1980PASJ...32...41S}, and its variants) requires the positivity of certain functions, we also determine the conditions for its applicability.

The paper is organized as follows: in Sections
\ref{sec:two_int_models} and \ref{sec:gen_sys}, after reviewing the
JEs for two-integral axisymmetric systems, the general solutions with
the $b$-ansatz, relating the vertical ($\sigma_z^2$) and the radial
($\sigma_R^2$) velocity dispersions, are derived, and cast in a form
suitable for the successive investigation. In Section
\ref{sec:phys_cons_bans}, we determine the constraints on the ansatz
to assure the positivity of $\vphiqm$ and to use the Satoh
decomposition and its generalizations. In Sections
\ref{sec:one_comp_mods} and \ref{sec:mnb_model}, we apply the
procedure to some well-known galaxy models, for which a fully
analytical treatment is possible: the Miyamoto \& Nagai and Satoh
disks, and the Binney logarithmic halo. A more realistic two-component
model, made by a stellar Miyamoto \& Nagai disk embedded in a dark
matter Binney halo, is also discussed. In Section \ref{sec:disc_conc},
the main results are summarized.

\section{Two-integral systems}\label{sec:two_int_models}

As well known, the JEs for an axisymmetric stellar system described by a two-integral DF $f=f(E, J_z)$, where $E$ and $J_z$ are respectively the orbital energy and the axial angular momentum, are

\begin{equation}
\begin{dcases}
    \frac{\partial\rhos\sigma_z^2}{\partial z}=-\rhos\frac{\partial\Phi}{\partial z}, \\\\
    \frac{\partial\rhos\sigma_z^2}{\partial R}-\frac{\rhos\Delta}{R}=-\rhos\frac{\partial\Phi}{\partial R},\quad \Delta \equiv \overline{v_{\varphi}^2}-\sigma_z^2,
\end{dcases}
\label{eq:2int_jeans}
\end{equation}

\noindent
(see e.g. BT08; C21): $\rhos(R,z)$ and $\Phi(R,z)$ are the stellar density distribution and the total (e.g., stars, plus dark matter) gravitational potential. Standard cylindrical coordinates are used, $\sigma_z^2$ is the vertical velocity dispersion, and a bar over a symbol indicates the average over velocity in phase-space. In particular, $\vphiqm = \vphim^2 + \sigma_{\varphi}^2$, where $\vphim$ is the streaming velocity field in the azimuthal direction, and $\sigma_{\varphi}^2$ is the azimuthal velocity dispersion. Finally, the only non-vanishing ordered velocity field is $\vphim$, while $\vrm = \vzm = 0$; moreover, $\sigma_R^2 = \sigma_z^2$, and all the mixed components of the velocity dispersion tensor vanish.

The solutions of equation (\ref{eq:2int_jeans}) with null boundary conditions at infinity are

\begin{equation}
\begin{dcases}
    \rhos\sigma_z^2 = \int_z^\infty \rhos\frac{\partial\Phi}{\partial z'}dz',\\\\
    \frac{\rhos\Delta}{R} = \frac{\partial\rhos\sigma_z^2}{\partial R}+\rhos\frac{\partial\Phi}{\partial R} = \left[\rhos,\Phi\right],
\end{dcases}
\label{eq:2int_JE_sol}
\end{equation}

\noindent
where

\begin{equation}
    \left[\rhos,\Phi\right] \equiv \int_z^{\infty}\!\left(\frac{\partial\rhos}{\partial R}\frac{\partial\Phi}{\partial z'}-\frac{\partial\rhos}{\partial z'}\frac{\partial\Phi}{\partial R}\right)\!dz',
    \label{eq:jeans_rad_comm}
\end{equation}

\noindent
is a commutator-like operator (see e.g. \citealt{Barnabe2006}; see also C21 and references therein). From the second of equation (\ref{eq:2int_JE_sol}), $\Delta$ can be obtained or by differentiation of $\rhos \sigma_z^2$, or by integration of the commutator, as in equation (\ref{eq:jeans_rad_comm}). This latter approach is to be preferred, as it usually reveals important properties of the solutions that are not apparent in the first approach based on differentiation: for example, it is immediate to show that the commutator vanishes for a spherical density $\rhos(r)$ in a spherical total potential $\Phi(r)$ and that, in the case of ellipsoidal densities in ellipsoidal potentials, the commutator is everywhere positive (negative) when the density shape is flatter (rounder) than the potential (see e.g. C21). Of course, when using the commutator for the computation of $\Delta$, the radial derivative of $\rhos \sigma_z^2$ is obtained as a byproduct:

\begin{equation}
    D \equiv \frac{\partial\rhos\sigma_z^2}{\partial R} = \left[\rhos,\Phi\right] - \rhos\frac{\partial \Phi}{\partial R},
    \label{eq:jeans_identity}
\end{equation}

\noindent
an identity we will use in the following. Once $\sigma_z^2$ and $\Delta$ are known, one has 

\begin{equation}
    \vphiqm = \Delta + \sigma_z^2 = \frac{\left[R \rhos, \Phi \right]}{\rhos},
    \label{eq:vphiqm_2int_comm}
\end{equation}

\noindent
where the second identity involving again a commutator can be easily proved from equations (\ref{eq:2int_JE_sol}) and (\ref{eq:jeans_rad_comm}).

Notice that a model with $\vphiqm < 0$ somewhere is certainly physically inconsistent, but $\vphiqm \ge 0$ everywhere is not a sufficient condition for consistency: as well known, there are models with \lq\lq acceptable" solutions of the JEs and a negative (unphysical) DF (see e.g. \citealt{Ciotti1992MNRAS.255..561C}; see also Chapter 14 in C21, and references therein).

\subsection{The Satoh $k$-decomposition}\label{sec:satoh_dec}
 
For axisymmetric models with $\Delta \ge 0$ everywhere, \citet{1980PASJ...32...41S} introduced the widely used \textit{$k$-decomposition} of $\vphiqm$:

\begin{equation}
     \vphim = k \sqrt{\Delta}, \;\; \sigma_{\varphi}^2 = \sigma_z^2 + (1-k^2) \Delta,
     \label{eq:satoh_dec}
\end{equation}

\noindent
where $k$ is constant with $0 \le k^2 \le 1$. If $k=0$, then $\vphim = 0$ and no net rotation is present, while, if $k^2 = 1$, then $\sigma_{\varphi}^2 = \sigma_z^2$ and the system is an isotropic rotator. Moreover, if $\Delta = 0$ (as for spherical models), then no ordered rotation is possible, and the system is isotropic independently of $k$. 

As shown in \citet{Ciotti1996MNRAS.279..240C}, the original Satoh decomposition can be easily generalized to assume a spatially dependent $k(R,z)$, provided that

\begin{equation}
    k^2(R,z) \le k_\text{M}^2(R,z) \equiv \frac{\vphiqm}{\Delta},
    \label{eq:satoh_dec_km}
\end{equation}

\noindent
where the upper limit $k_\text{M}^2$ corresponds to maximally rotating models with no net velocity dispersion in the azimuthal direction, and is then obtained from equation (\ref{eq:satoh_dec}) imposing $\sigma_{\varphi}^2 = 0$ everywhere; in this case, the density flattening is fully supported by the streaming velocity field $\vphim$. Of course, $k$ is not required to be positive, thus allowing the modelling of counter-rotating structures with negative $\vphim$ (see e.g. \citealt{Negri2013MmSAI..84..762N}).

In Appendix \ref{sec:lambda_dec}, we discuss the most general decomposition for $\vphiqm$, which holds also for models with $\Delta < 0$.

\section{More general systems}\label{sec:gen_sys}

Having assessed the classical case of two-integral axisymmetric systems, we now turn to the focus of the paper, i.e. the study of the properties of the kinematical fields 
associated with more general ansatz distinguishing between $\sigma_z^2$ and $\sigma_R^2$, and so implicitly based on a DF with a third integral in addition to $E$ 
and $J_z$. If the third integral is an even function of $v_z$ and $v_R$, the first of equation (\ref{eq:2int_jeans}) remains unchanged, while the second becomes 
(\citealt{Cappellari2008MNRAS.390...71C}, C21):

 \begin{equation}
        \frac{\partial\rhos\sigma_R^2}{\partial R}-\frac{\rhos\Deltac}{R}=-\rhos\frac{\partial\Phi}{\partial R}, \quad \Deltac \equiv \vphiqm - \sigma_R^2.
 \label{eq:jeans_rad_anis}
 \end{equation}

\noindent
Notice that, in addition to $\sigma_z^2$, all quantities depending on $\rhos$ and $\Phi$, such as the commutator $\left[ \rhos, \Phi \right]$ and the function $D$ in equation (\ref{eq:jeans_identity}), remain the same as in the two-integral case.

Different ansatz can be introduced to solve equation
(\ref{eq:jeans_rad_anis}), which is independent of the vertical
velocity dispersion, and contains the two unknown functions $\vphiqm$
and $\sigma_R^2$. In the following, we study in detail the ``$b$-\textit{ansatz}", relating $\sigma_R^2$ with $\sigma_z^2$.
Introduced by \citet{Cappellari2008MNRAS.390...71C}, it was
adopted for the Jeans Anisotropic Modelling method (JAM; see also  
\citealt{Cappellari2020MNRAS.494.4819C}), that is 
widely used to reproduce the properties of observed galaxies (e.g., \citealt{Cappellari2013MNRAS.432.1862C}, \citealt{2016MNRAS.462.4001Z}, \citealt{2020MNRAS.496.1857L}, \citealt{2021ApJ...916..112N}, \citealt{2024ApJ...960..110S}). This solution of the JEs, implying the alignment of the velocity ellipsoid with the cylindrical coordinates, was presented as an efficient modelling able to capture the main properties of the velocity ellipsoid inferred from extensive three-integrals Schwarzschild's modelling of integral-field stellar kinematics, under the mass-follows-light hypothesis (\citealt{Cappellari2008MNRAS.390...71C}). The main motivation for the adoption of the \lq\lq$b$-\textit{ansatz}", in fact, is that it allows to model adequately the observations, in particular the integral-field spectroscopy of axisymmetric galaxies classified as regular rotators and stellar disks (see Sections 2.3 and 2.4 of \citealt{Cappellari2008MNRAS.390...71C}; \citealt{Cappellari2016ARA&A..54..597C}). Of course, on the
theoretical side, the $b$-ansatz is not the only one possible, and in Appendix \ref{sec:mu_ans} we present the \lq\lq $\mu$-\textit{ansatz}", leading to a nice closure of equation (\ref{eq:jeans_rad_anis}).

\subsection{The b-ansatz}\label{sec:b_ans}

In the \lq\lq\textit{$b$-ansatz}", the unknown $\sigma_R^2$ is linked to $\sigma_z^2$ through the choice of the function $b(R,z) \ge 0$, as

\begin{equation}
    \sigma_R^2 = b(R,z) \sigma_z^2.
    \label{eq:b_ans_def}
\end{equation}

\noindent
When $b=0$, the system has no radial velocity dispersion, while $b=1$ gives the two-integral case.

Inserting equation (\ref{eq:b_ans_def}) in equation (\ref{eq:jeans_rad_anis}), recalling the definition of $D$ in equation (\ref{eq:jeans_identity}), and solving for $\Deltac$, we obtain:

\begin{equation}
\begin{split}
\frac{\rhos\Deltac}{R} -\rhos\sigma_z^2\frac{\partial b}{\partial R} =\, &bD + \rhos\frac{\partial\Phi}{\partial R} =\\
&b\left[\rhos,\Phi\right] + (1-b)\rhos\frac{\partial\Phi}{\partial R},
\end{split}
\label{eq:b_ans_delta}
\end{equation}

\noindent
where in the last expression we have used again equation (\ref{eq:jeans_identity}). Note how $b$ multiplies functions that are \textit{independent} of the adopted ansatz.

Once $\Deltac$ is known, if we restrict to a $b$ function that depends only on $z$ (or to a constant $b$, a special case commonly used), $\partial b/\partial R = 0$ on the l.h.s. of equation (\ref{eq:b_ans_delta}), and

\begin{equation}
\begin{split}
\vphiqm = \Deltac + b \sigma_z^2 =\, &bB + R\frac{\partial\Phi}{\partial R} =\\
&b \frac{\left[ R \rhos, \Phi \right]}{\rhos} + (1-b)R\frac{\partial\Phi}{\partial R},
\end{split}
\label{eq:vphi2m_anis}
\end{equation}

\noindent
where from equation (\ref{eq:jeans_identity})

\begin{equation}
    B \equiv \frac{R\,D}{\rhos} + \sigma_z^2 = \frac{\left[ R \rhos, \Phi \right]}{\rhos} - R \frac{\partial\Phi}{\partial R}.
    \label{eq:b_def}
\end{equation}

\noindent
Only $b(z)$ functions leading to $\vphiqm \ge 0$ everywhere are physically acceptable, and Section \ref{sec:cons_b_ans} deals with this request.

Finally, from equation (\ref{eq:vphi2m_anis}), we obtain:

\begin{equation}
    \Delta = \vphiqm - \sigma_z^2 = \Deltac + (b-1)\sigma_z^2 = bB + C,
    \label{eq:delta_bans}
\end{equation}

\noindent
where from equations (\ref{eq:jeans_identity}) and (\ref{eq:b_def})

\begin{equation}
    C \equiv R\frac{\partial\Phi}{\partial R} - \sigma_z^2 = \frac{R \left[ \rhos,\Phi \right]}{\rhos} - B.
    \label{eq:c_function_def}
\end{equation}

\noindent
From equation (\ref{eq:delta_bans}), if $\Deltac \ge 0$ and $b(z) \ge 1$, then $\Delta \ge 0$, and if $\Deltac < 0$ and $b(z) \le 1$, then $\Delta < 0$. A complete discussion on the conditions for the positivity of $\Deltac$ and $\Delta$ is given in Sections \ref{sec:pos_delta} and \ref{sec:pos_delta_C}. Notice that the values of $\Delta$ in the $b$-ansatz are not the same of the two-integral case, because, with the introduction of the ansatz, $\sigma_z^2$ remains unaltered, but $\vphiqm$ changes. Also important, the functions $D$, $B$ and $C$ are ansatz-independent.

Models with $\Deltac \ge 0$ allow for a $k$-decomposition of $\vphiqm$ similar to that in equation (\ref{eq:satoh_dec}), i.e. (\citealt{Cappellari2008MNRAS.390...71C}):

\begin{equation}
    \vphim = k \sqrt{\Deltac}, \; \; \sigma_{\varphi}^2 = \sigma_R^2 + (1-k^2)\Deltac.
    \label{eq:satoh_dec_bans}
\end{equation}

\noindent
Whenever $k^2<1$, one has $\sigma_{\varphi}^2 > \sigma_R^2$, i.e. the orbital anisotropy is tangential; radial anisotropy ($\sigma_{\varphi}^2 < \sigma_R^2$) requires $k^2>1$. Of course, the decomposition in equation (\ref{eq:satoh_dec}) can also be applied, where now $\Delta$ is given in equation (\ref{eq:delta_bans}), provided that $\Delta \ge 0$ everywhere.

A different decomposition for the azimuthal motions can be given by choosing the parameter $\gamma$, introduced in \cite{Cappellari2007MNRAS.379..418C}:

\begin{equation}
    \gamma = 1 - \frac{\sigma_{\varphi}^2}{\sigma_R^2},\quad \vphim^2 = \Deltac + \gamma \sigma_R^2,
    \label{eq:gamma_def}
\end{equation}

\noindent
where clearly

\begin{equation}
    -\frac{\Deltac}{\sigma_R^2} \le \gamma \le 1,
    \label{eq:gamma_cons}
\end{equation}

\noindent
with the first inequality required for $\vphim^2 \ge 0$. Positive $\gamma$ values correspond to radial anisotropy, while $\gamma < 0$ gives tangential anisotropy;
in two-integral systems, $\gamma = 0$ corresponds to the isotropic case. Notice that, from equation (\ref{eq:gamma_cons}), $\gamma$ can be either positive or negative if $\Deltac > 0$; instead, if $\Deltac < 0$, $\gamma$ cannot be taken negative, i.e. the velocity dispersion tensor can only be radially anisotropic. In Section \ref{sec:disc_conc}, some findings based on the use of this decomposition to interpret recent observations (\citealt{Wang2021MNRAS.500L..27W}) are discussed in light of our analysis.

\section{Physical Constraints on the Ansatz}\label{sec:phys_cons_bans}

Equation (\ref{eq:jeans_rad_anis}) can be solved only with the introduction of some closure, as those presented in Section \ref{sec:b_ans} or Appendix \ref{sec:mu_ans}. However, arbitrary ansatz functions can lead to unphysical solutions of the JEs, such as negative values of $\vphiqm$. Therefore, it would be useful to know in advance (i.e., before solving the equations) what constraints must be imposed on the ansatz function in order to avoid unphysical solutions, and also how specific choices of the ansatz functions affect the properties of the solutions in different regions of space. This Section is dedicated to these problems.

In practice, being $\sigma_z^2 \ge 0$ independent of the ansatz, the positivity of $\sigma_R^2$ is guaranteed by equation (\ref{eq:b_ans_def}); therefore, we focus on the positivity of $\vphiqm$ only. Information on the sign of $\Delta$ and $\Deltac$ is also relevant for the modelling because, although not directly related to model consistency, their positivity is needed to apply a decomposition of $\vphiqm$ as those in equations (\ref{eq:satoh_dec}) and (\ref{eq:satoh_dec_bans}).

In the following, we restrict for simplicity to models with $\partial \Phi/\partial R \ge 0$ everywhere (the most common situation), and we limit to consider $b = b(z)$; finally, we indicate with $\mathrm{P} = (R,z)$ a generic point in the ($R,z$)-plane.

\subsection{Positivity of $\vphiqm$: the $\mathscr{B}$ region}\label{sec:cons_b_ans}

We discuss here the positivity of $\vphiqm$, a condition necessary to model consistency. Introducing the sets

\begin{equation}
    \mathscr{B}^{\pm} = \left\{(R,z): B \gtrless 0 \right\},\quad \bzero = \left\{(R,z): B = 0 \right\},
    \label{eq:b_sets_def}
\end{equation}

\noindent
equation (\ref{eq:vphi2m_anis}) shows that, independently of $b(z)$, $\vphiqm \ge 0$ over the region $\bpzero = \bplus \cup \bzero$. Instead, the behaviour of $B$ over $\bminus$ constrains the possible choices of $b(z)$, as we now describe. We indicate with $\projz(\bpm)$ the \textit{projection} of $\bpm$ on the $z$-axis, and with $\bzpm = \left\{R: B \gtrless 0, z \in \projz(\bpm) \right\}$ the \textit{radial section} at fixed $z$ of $\bpm$ (see Figure \ref{fig:set_ex} for a qualitative illustration). As we restrict to systems with a reflection symmetry about the equatorial plane, without loss of generality, in the following we limit the discussion to $z \ge 0$. From equation (\ref{eq:vphi2m_anis}), the condition $\vphiqm \ge 0$ over $\bminus$ is guaranteed provided that, at each $z \in \projz(\bminus)$,

\begin{equation}
\begin{split}
       b(z) &\le \bM(z) \equiv \min_{\bzminus} \frac{R}{\lvert B \rvert} \frac{\partial\Phi}{\partial R} =\\
       & \min_{\bzminus} \left( 1 + \frac{R\left[ \rhos, \Phi \right] + \rhos\sigma_z^2}{\rhos \lvert B \rvert} \right),\quad z \in \projz(\bminus),
       \label{eq:bmz_criterion}
\end{split}
\end{equation}

\noindent
where the last equality has been obtained using equation (\ref{eq:c_function_def}). Notice that the independence of $b$ from $R$ implies that the upper limit $\bM(z)$, determined over $\bzminus$, applies also to points in the complementary section $\bzplus$: therefore, the function $\bM(z)$ provides a constraint over the whole \textit{rectangular strip} defined by $R \ge 0$ and $z \in \projz(\bminus)$. From now on, we will refer to this region as to \lq\lq the strip $\projz(\bminus)$". In the special case of a constant $b$, the condition in equation (\ref{eq:bmz_criterion}) reduces to $b \le \bM$, where $\bM$ is the minimum of $\bM(z)$.

\begin{figure}
\centering
    \includegraphics[width = \linewidth]{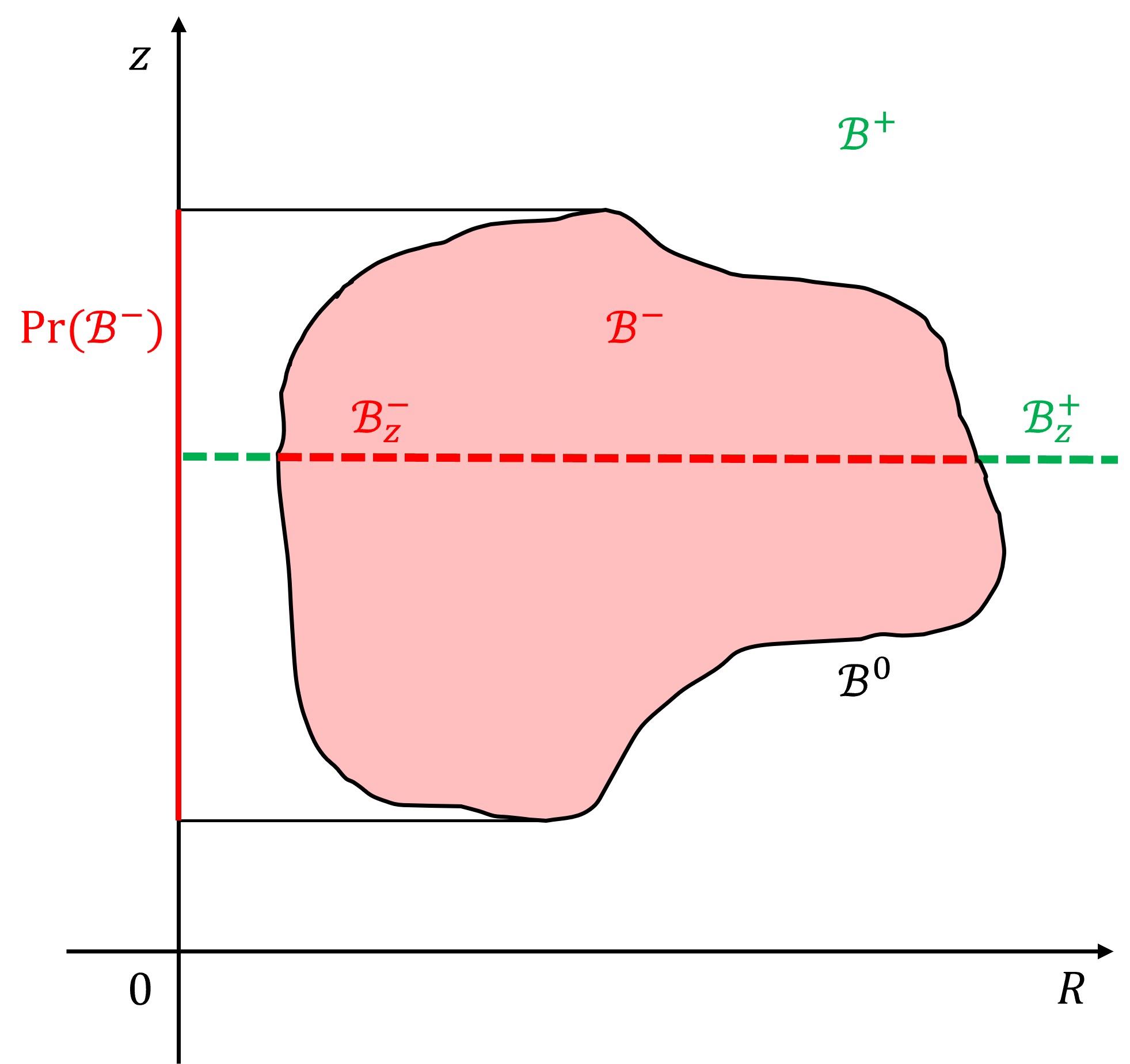}
    \caption{Cartoon of a fictitious $\bminus$ region (red), together with its projection on the $z$-axis $\projz(\bminus)$, a radial section $\bzminus$ (red dashed line), and its complement $\bzplus$ (green dashed line). From equation (\ref{eq:b_sets_def}), $\bplus \cup \bminus \cup \bzero$ fully covers the ($R,z$)-plane.}
    \label{fig:set_ex}
\end{figure}

\subsection{Positivity of $\Deltac$: the $\mathscr{D}$ region}\label{sec:pos_delta}

As seen in Section \ref{sec:b_ans}, the necessary condition to apply the $k$-decomposition to $\vphiqm$ is to have $\Deltac \ge 0$ or $\Delta \ge 0$. Introducing the sets

\begin{equation}
    \mathscr{D}^{\pm} = \left\{(R,z) : D \gtrless 0 \right\},\quad \dzero = \left\{(R,z) : D = 0 \right\},
    \label{eq:dpm}
\end{equation}

\noindent
where $D$ is defined in equation (\ref{eq:jeans_identity}), the first of equation (\ref{eq:b_ans_delta}) (where now $\partial b/\partial R = 0$) shows that, independently of $b(z)$, $\Deltac \ge 0$ over the region $\dpzero = \dplus \cup \dzero$. Instead, the condition $\Deltac \ge 0$ over $\dminus$ requires

\begin{equation}
\begin{split}
     b(z) &\le b_0(z) \equiv \min_{\dzminus}\frac{\rhos}{\lvert D \rvert}\frac{\partial\Phi}{\partial R} =\\
     &= \min_{\dzminus} \left( 1 + \frac{\left[ \rhos, \Phi \right]}{\lvert D \rvert} \right),\quad z \in \projz(\dminus),
    \label{eq:b0z_criterion}
\end{split}
\end{equation}

\noindent
where $\projz(\dpm)$ is the projection of $\dpm$ on the $z$-axis, $\dzpm = \left\{R: D \gtrless 0, z \in \projz(\dpm) \right\}$, and the last equality follows from equation (\ref{eq:jeans_identity}). Due to the independence of $b$ from $R$, the condition in equation (\ref{eq:b0z_criterion}) must be verified over the whole strip $\projz(\dminus)$. In the special case of a constant $b$, then $\Deltac \ge 0$ for $b \le b_0$, where $b_0$ is the minimum of $b_0(z)$.

Notice that from Theorem \ref{thm:0} the strip $\projz(\bminus)$ is contained in the strip $\projz(\dminus)$, and in this common region $b_0(z) \le \bM(z)$, as proved in Theorem \ref{thm:4}.

\subsection{Positivity of $\Delta$: the $\mathscr{C}$ region}\label{sec:pos_delta_C}

The determination of the positivity of $\Delta$ as a function of $b(z)$ is more complicated, depending on the sign of both the $B$ and $C$ functions in equation (\ref{eq:delta_bans}). As done above, we introduce the sets

\begin{equation}
    \mathscr{C}^{\pm} = \left\{(R,z): C \gtrless 0 \right\},\quad \czero = \left\{(R,z): C = 0 \right\},
    \label{eq:cpm}
\end{equation}

\noindent
and $\czpm = \left\{R : C \gtrless 0, z \in \projz(\cpm) \right\}$; also, $\cpzero = \cplus \cup \czero$. The sign of $\Delta$ at a point $\mathrm{P}$ depends on its position in the ($R,z$)-plane, as follows:

\begin{enumerate}
    \item $\Delta \ge 0$ independently of $b(z)$ over $\bpzero \cap \cpzero$. Moreover, $\Delta \ge 0$ over $\bplus \cap \cminus$ for

\begin{equation}
    \begin{split}
    b(z) &\ge b_1(z) \equiv \max_{\bzplus \cap \czminus} \frac{\lvert C \rvert}{B} =\\ 
    & \max_{\bzplus \cap \czminus} \left( 1 - \frac{R \left[\rhos, \Phi \right]}{\rhos B} \right),\;\; z \in \projz(\bplus \cap \cminus),\label{eq:b1_def}
\end{split}
    \end{equation}

    \noindent
    and over $\bminus \cap \cplus$ for

    \begin{equation}
        \begin{split}
    b(z) &\le b_2(z) \equiv \min_{\bzminus \cap \czplus} \frac{C}{\lvert B \rvert} =\\
    & \min_{\bzminus \cap \czplus} \left( 1 + \frac{R \left[\rhos, \Phi \right]}{\rhos \lvert B \rvert} \right),\;\; z \in \projz(\bminus \cap \cplus).\label{eq:b2_def}
\end{split}
    \end{equation}
    
    \item $\Delta < 0$ independently of $b(z)$ over $(\bminus \cap \cminus) \cup (\bzero \cap \cminus) \cup (\bminus \cap \czero)$\footnote{This region is the complementary of the region in (i), minus the regions in equations (\ref{eq:b1_def}) and (\ref{eq:b2_def}), as can be proved with the De Morgan's law.}. Moreover, $\Delta < 0$ over $\bplus \cap \cminus$ for $b(z) < b_1(z)$, and over $\bminus \cap \cplus$ for $b(z) > b_2(z)$.
\end{enumerate}

\noindent
For a spatially constant $b$, the previous inequalities hold replacing $b_1(z)$ and $b_2(z)$ with $b_1$ and $b_2$, which are respectively the maximum of $b_1(z)$, and the minimum of $b_2(z)$.

Summarizing, independently of $b(z)$, $\vphiqm \ge 0$ over $\bpzero$, $\Deltac \ge 0$ over $\dpzero$, and $\Delta \ge 0$ over $\bpzero \cap \cpzero$, and $\Delta < 0$ over $(\bminus \cap \cminus) \cup (\bzero \cap \cminus) \cup (\bminus \cap \czero)$. Instead, equations (\ref{eq:bmz_criterion}), (\ref{eq:b0z_criterion}), (\ref{eq:b1_def}) and (\ref{eq:b2_def}) set the constraints on $b(z)$ required for the positivity of $\vphiqm$ over $\bminus$, $\Deltac$ over $\dminus$, and $\Delta$ over the intersection regions $\bplus \cap \cminus$ and $\bminus \cap \cplus$.

\subsection{Some general considerations}\label{sec:gen_cons}

The properties of the kinematical fields resulting from a given choice for $b(z)$ are determined by the relations between the ansatz-independent regions $\bpm$, $\cpm$ and $\dpm$ and the functions defined over them. Each region is made by two disjoint parts fully covering the ($R,z$)-plane, thus the determination of just $\bminus$, $\cminus$ and $\dminus$ suffices for investigating a given model. In Appendix \ref{app:tech_res}, some general results about the relative positions of the regions and the constraints on $b(z)$ are presented. In particular, it is shown that $\bminus$ is always contained in $\dminus$. Moreover, if $\left[\rhos,\Phi\right] > 0$ over the whole ($R,z$)-plane (as for oblate self-gravitating models), then $\bminus \subseteq \cplus$ and $\cminus \subseteq \bplus$; if $\left[\rhos,\Phi\right] = 0$ everywhere (as in spherically symmetric models), then $\bminus = \cplus$, $\bplus = \cminus$, and $\dmzero$ coincides with the whole ($R,z$)-plane; if $\left[\rhos,\Phi\right] < 0$ (as for prolate self-gravitating models), then $\bplus \subseteq \cminus$, $\cplus \subseteq \bminus$, and $\dminus$ coincides with the whole ($R,z$)-plane. Finally, from Theorems \ref{thm:b1} and \ref{thm:b0}, it follows that for models with $\left[ \rhos, \Phi \right] \ge 0$ everywhere (a quite common situation), one has:

\begin{equation}
\begin{dcases}
    b_1(z) \le 1,\quad \forall z \in \projz(\cminus),\\\\
    1 \le b_0(z) \le b_2(z) \le \bM(z),\quad \forall z \in \projz(\bminus).
\end{dcases}
\label{eq:bz_hierarchy}
\end{equation}

\noindent
In this case, from the considerations above, $\bminus$ and $\cminus$ are disjoint, but this does not mean that the two projections in equation (\ref{eq:bz_hierarchy}) are disjoint too: actually, they may coincide with the $z$-axis. In the special case of $\left[ \rhos, \Phi \right] = 0$ everywhere, in equation (\ref{eq:bz_hierarchy}) we would have $b_1 = b_0 = b_2 = 1$.

We conclude with an example of the qualitative effects of an increase of $b$ (for simplicity assumed as spatially constant) on $\sigma_R^2$, $\vphiqm$, $\Deltac$ and $\Delta$, in the various regions. Suppose a given model is assigned, and the regions $\bpm$, $\cpm$ and $\dpm$ have been determined; furthermore, suppose $\left[ \rhos, \Phi \right] \ge 0$ everywhere, so that equation (\ref{eq:bz_hierarchy}) holds. We start considering the case of $b = 0$ (i.e. $\sigma_R^2 = 0$ everywhere): as $0$ is smaller than $b_0$ and $\bM$, then $\vphiqm$ and $\Deltac$ are positive everywhere, while $\Delta < 0$ somewhere in $\cminus$ (because $0 < b_1)$, and $\Delta \ge 0$ everywhere in $\cplus$ (because $0 < b_2$). For increasing $b$, $\sigma_R^2$ increases everywhere, $\vphiqm$ decreases over $\bminus$ and increases over $\bplus$, while $\Deltac$ decreases over $\dminus$ and increases over $\dplus$. As $\bminus \subseteq \dminus$, over $\bminus$ one has that $\Deltac$ decreases because $\vphiqm$ decreases and $\sigma_R^2$ increases; over the remaining part of $\dminus$ (i.e. $\bplus \cap \dminus$), instead, $\Deltac$ decreases, and so $\sigma_R^2$ increases more than $\vphiqm$. Increasing further $b$ above $b_1$, $\Delta$ becomes positive over $\cminus$, then when $b > b_0$, $\Deltac$ becomes negative somewhere over $\dminus$, and for $b> b_2$, $\Delta$ becomes negative somewhere over $\bminus$; finally, for $b > \bM$, $\vphiqm < 0$ somewhere in $\bminus$, and the model becomes unphysical. These behaviours are summarized in the table in Figure \ref{fig:fields_table}.

\begin{figure}
\centering
    \includegraphics[width = \linewidth]{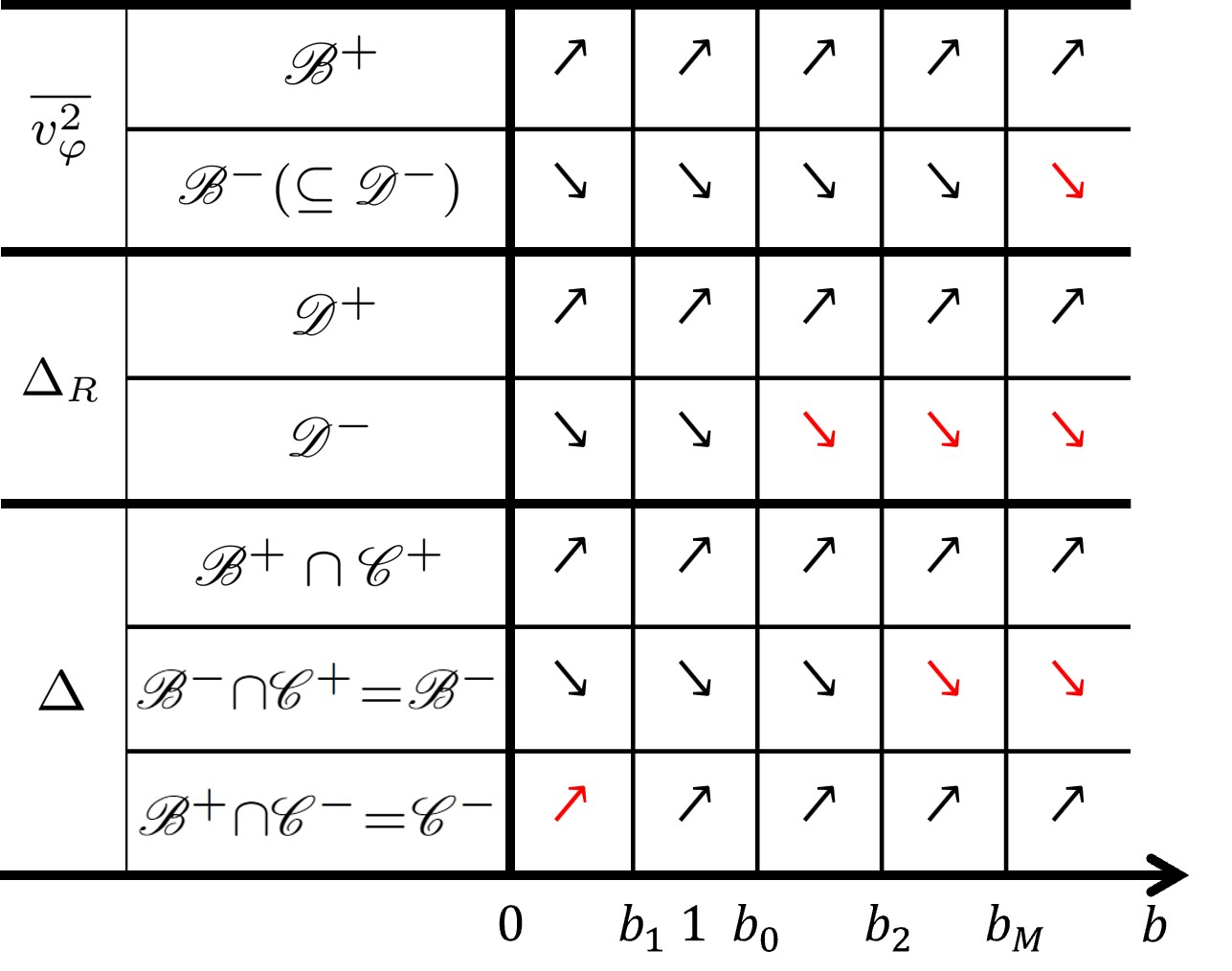}
    \caption{Scheme representing the behaviour, as a function of $b$, of the kinematical fields $\vphiqm$, $\Deltac$ and $\Delta$ over the regions $\bpm$, $\cpm$ and $\dpm$, for a model with $\left[\rho,\Phi\right] > 0$. Each row refers to the indicated region; arrows point upwards if the field is increasing, and downwards if the field is decreasing for increasing $b$; black arrows indicate a positive value, and red arrows a negative one. For $b > \bM$, the model is unphysical.}
    \label{fig:fields_table}
\end{figure}

\section{One-Component Models} \label{sec:one_comp_mods}

In light of the results in Section \ref{sec:phys_cons_bans}, we now explore the behaviour of some well-known galaxy models that allow for an almost complete analytical treatment, i.e. the \citet{1975PASJ...27..533M} (hereafter MN) disk, the \citet{1980PASJ...32...41S} disk, and the Binney logarithmic halo (BT08). In Section \ref{sec:mnb_model}, we then consider the case of the two-component model made by a stellar MN disk embedded in a dark matter Binney logarithmic halo.

\subsection{The Miyamoto \& Nagai disk}\label{sec:MN_model}

The potential-density pair of the MN disk of total mass $M_*$, and scale lengths $a_*$ and $b_*$, can be written as:

\begin{equation}
\begin{dcases}
    \phis = -\frac{GM_*}{b_* \xi}, \quad \xi \equiv \sqrt{R^2 + (s+\zeta)^2},\\
    \rhos = \frac{M_*}{b_*^3}\frac{s \xi^2 + 3\zeta (s+\zeta)^2}{4\pi\zeta^3 \xi^5}, 
\label{eq:MN_pd_pair}
\end{dcases}
\end{equation}

\noindent
where $\zeta \equiv \sqrt{1+z^2}$, and $s \equiv a_*/b_*$ measures the flattening of the disc. For $a_* = 0$, the MN model reduces to the \citet{1911MNRAS..71..460P} sphere, and for $b_* = 0$ to the razor-thin Kuzmin disc (\citealt{kuzmin1956model}, \citealt{toomre1963ApJ...138..385T}). In the formulae above and in this Section, $R$ and $z$ are assumed to be normalized to $b_* \ne 0$.

From equations (\ref{eq:2int_JE_sol}) and (\ref{eq:jeans_rad_comm}), we have:

\begin{equation}
    \rhos\sigma_z^2 = \frac{G M_*^2}{b_*^4}\frac{(s+\zeta)^2}{8\pi \zeta^2 \xi^6},\quad \left[\rhos,\Phi_*\right] = \frac{G M_*^2}{b_*^5}\frac{sR}{4\pi \zeta^3 \xi^6}.
    \label{eq:delta_MN}
\end{equation}

\noindent
In particular, $\left[\rho_*,\Phi_*\right] \ge 0$ everywhere, vanishing in the spherical case ($s=0$): therefore, equation (\ref{eq:bz_hierarchy}) and the considerations made at the end of Section \ref{sec:phys_cons_bans} apply. For reference, in the top panels of Figure \ref{fig:2int_1comp_panel}, we show the 2D maps of the ansatz-independent fields $\sigma_z^2$, $R\left[\rho_*,\Phi_*\right]/\rhos$ (i.e. the $\Delta$ field in the two-integral case), and $R \partial\Phi/\partial R$, for a disc flattening of $s = 1$.

\begin{figure*}
    \centering
    \includegraphics[width=\textwidth]{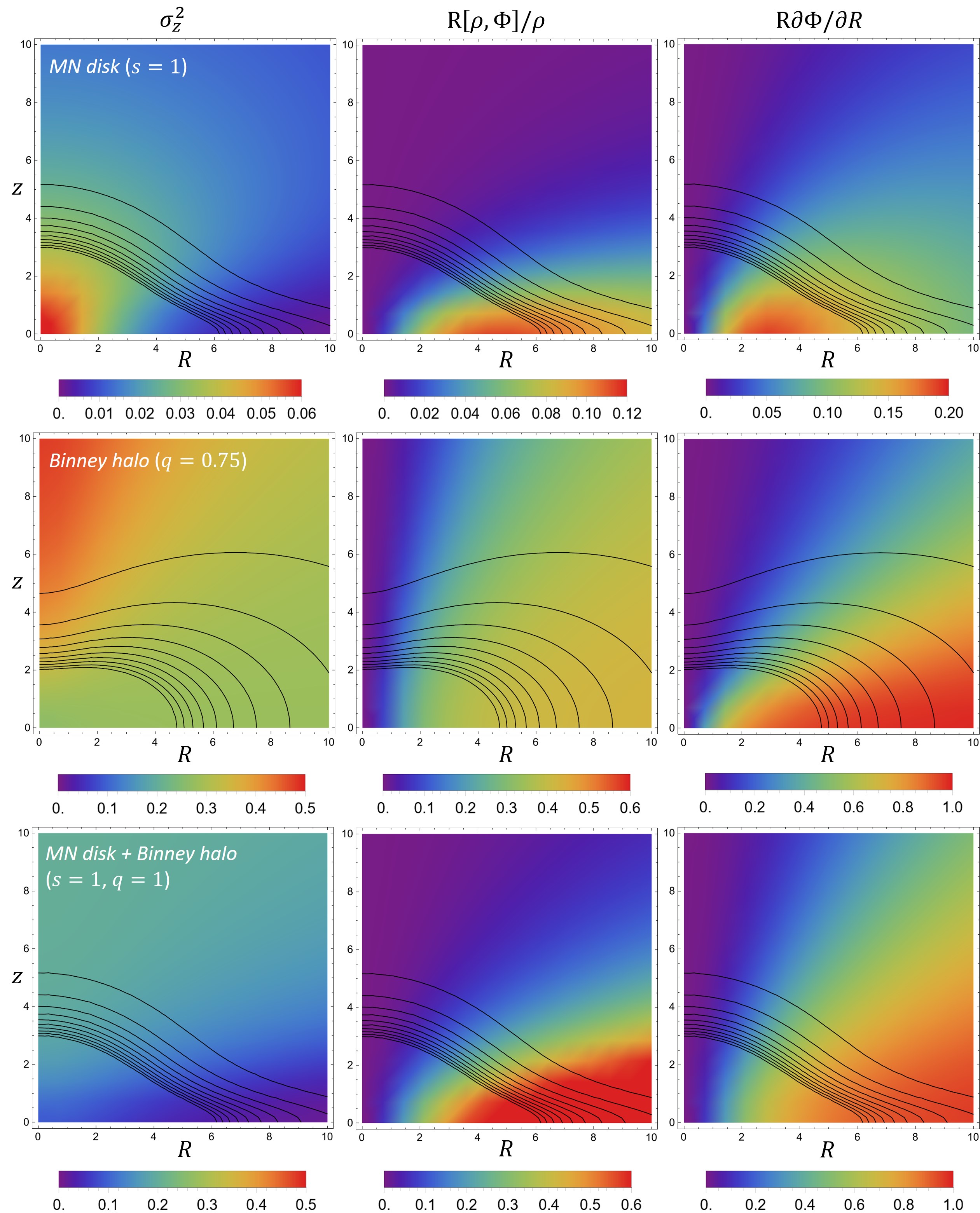}
    \caption{Maps in the ($R,z$)-plane of $\sigma_z^2$ (left column), $R\left[\rho,\Phi\right]/\rho$ (central column), and $R\partial\Phi/\partial R$ (right column) for the MN disk with $s=1$ (top row), for the Binney logarithmic halo with $q=0.75$ (central row), and for the MN stellar disk ($s = 1$) in a dominant ($\Phi = \phih$) and spherical ($q = 1$) Binney logarithmic dark matter halo with $\bh = 2$. The fields are in units of $\vh^2$ for the Binney halo, and of $GM_*/b_*$ for the other two models. Black solid lines show equally spaced isodensity contours.}
    \label{fig:2int_1comp_panel}
\end{figure*}

The first task to explore the model behaviour is the identification of the $\mathscr{B}$, $\mathscr{C}$ and $\mathscr{D}$ regions. Quite remarkably, equation (\ref{eq:delta_MN}) allows for an analytical expression of the $B$, $C$ and $D$ functions:

\begin{equation}
    \bminus = \left\{(R,z) : \sqrt{5}R - (s + \zeta) > 0 \right\},
    \label{eq:mn_b_exp}
\end{equation}

\noindent
i.e. $\bminus$ is the portion of the ($R,z$)-plane at the right of the hyperbola $\bzero$ with vertex $R=(1+s)/\sqrt{5}$ on the equatorial plane, and asymptotes $z =\nobreak \pm\sqrt{5} R$, so that $\projz(\bminus)$ coincides with the $z$-axis. In the top panels of Figure \ref{fig:MN_B_Reg}, the red line shows $\bzero$ for three representative values of $s$. Moreover, the radial derivative of the first of equation (\ref{eq:delta_MN}) shows that $D < 0$ over the whole space, and $D = 0$ for $R = 0$, i.e. $\dzero$ coincides with the $z$-axis; therefore, $\projz(\dminus)$ coincides with the $z$-axis, and $\dplus$ is empty. Finally, from the second of equation (\ref{eq:delta_bans}), we get

\begin{equation}
\begin{split}
  \cminus = \Big\{(R,z) :\, &2 s R^4 + (s + \zeta)^2 (2 s + 5 \zeta) R^2 -\\
  &(s+\zeta)^4 \zeta < 0\Big\}.
\end{split}
\label{eq:cmzero_MN}
\end{equation}

\noindent
The biquadratic in $R$ above has a positive discriminant, and a permanence and a variation of signs in the coefficients: it follows that $\cminus$ is the portion of the ($R,z$)-plane contained between the $z$-axis and the (positive) square root of the (positive) solution of the biquadratic. We do not report here the expression for $\czero$, which however can be determined without difficulty; $\projz(\cminus)$ coincides with the $z$-axis. In the top panels of Figure \ref{fig:MN_B_Reg}, the green line shows $\czero$ for three representative values of $s$. Again from Figure \ref{fig:MN_B_Reg}, it is apparent that $\bminus \subseteq \cplus$ and $\cminus \subseteq \bplus$, as expected from the general discussion in Section \ref{sec:phys_cons_bans}; notice also how the separation between $\bminus$ and $\cminus$ (i.e. the region $\bplus \cap \cplus$) becomes larger as $s$ increases. Moreover, in the spherical case ($s = 0$) we would obtain $\bzero = \czero$, $\bminus = \cplus$ and $\bplus = \cminus$.

\begin{figure*}
\centering
        \includegraphics[width=\linewidth]{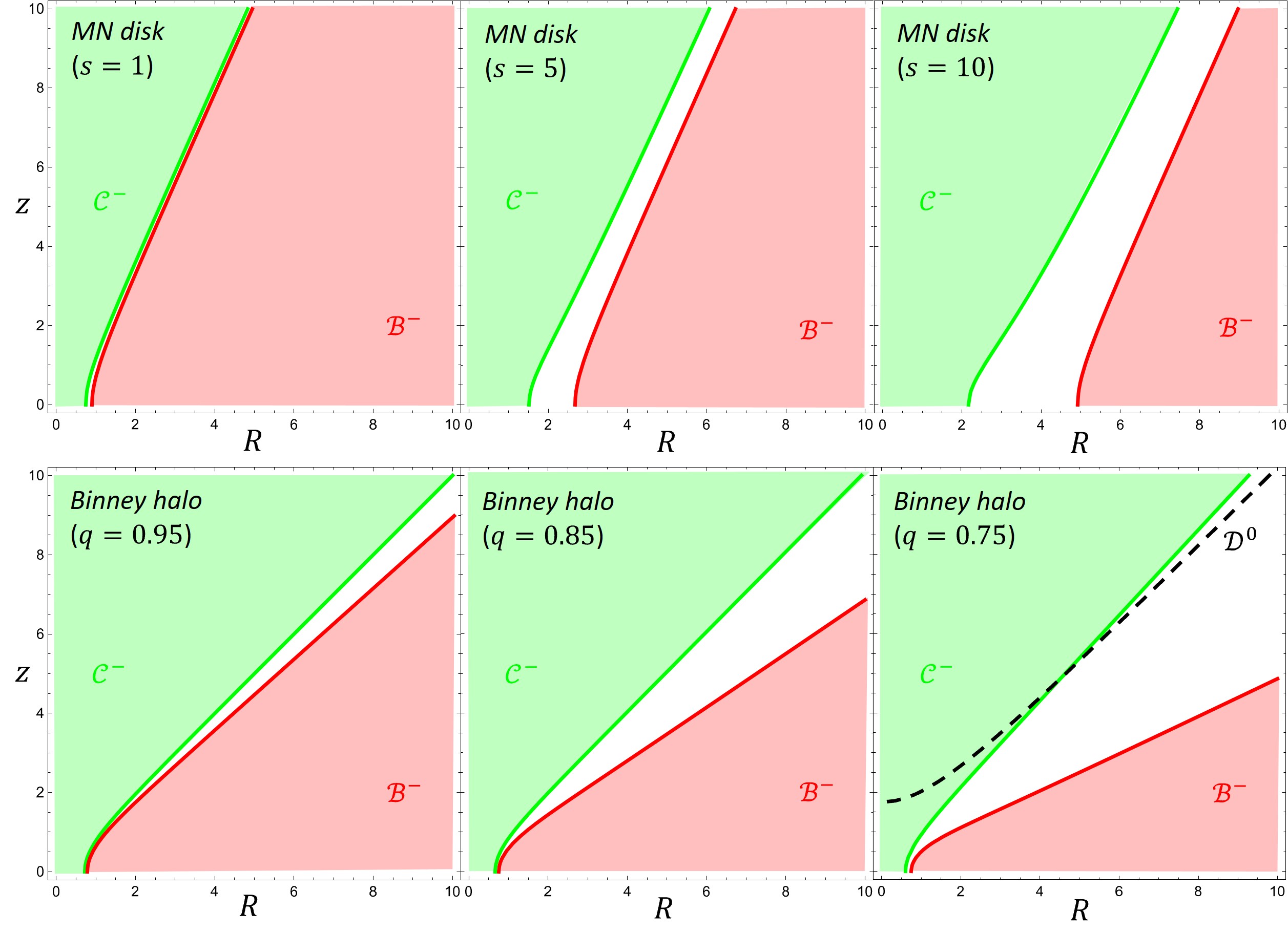}
        \caption{The ansatz-independent regions $\mathscr{B}$, $\mathscr{C}$, and $\mathscr{D}$ of the MN disk (top panels) for different values of $s$, and of the Binney logarithmic halo (bottom panels) for different values of $q$; both models become flatter from left to right. Red and green areas correspond to the $\bminus$ and $\cminus$ regions, respectively, while heavy red and green lines are the $\bzero$ and $\czero$ sets. Except for the Binney model with $q = 0.75$, where $\dminus$ lies below the $\dzero$ line (dashed), for all the other cases $\dminus$ fully covers the ($R,z$)-plane, with $\dzero$ coinciding with the $z$-axis. For all models, $\projz(\bminus)$ and $\projz(\cminus)$ coincide with the $z$-axis, and as $\left[\rho,\Phi\right] \ge 0$ everywhere, from Theorem \ref{thm:2} we have $\bminus \cap \cplus = \bminus$, and $\cminus \cap \bplus = \cminus$.}
\label{fig:MN_B_Reg}
\end{figure*}

We now determine the constraints on $b(z)$ for the positivity of $\vphiqm$, $\Deltac$ and $\Delta$. Remarkably, all the computations in equations (\ref{eq:bmz_criterion}), (\ref{eq:b0z_criterion}), (\ref{eq:b1_def}) and (\ref{eq:b2_def}) can be carried out explicitly, and we obtain

\begin{equation}
    \bM(z) = \frac{6}{5} + \frac{14 s + 4\sqrt{s\left(6s+15\zeta\right)}}{25\zeta},
    \label{eq:bmz_MN_exp2}
\end{equation}

\begin{equation}
    b_0(z) = 1 + \frac{s}{3\zeta},
    \label{eq:b0_MN}
\end{equation}

\begin{equation}
    b_1(z) = 1,\quad b_2(z) = 1 + \frac{(14+4\sqrt{6})s}{25\zeta}.
\end{equation}

\noindent
In the top-left panel of Figure \ref{fig:1comp_bprof}, we show the four functions for $s = 1$, where the chain of inequalities in equation (\ref{eq:bz_hierarchy}) is apparent, considering that for the MN disk $\projz(\bminus)$ and $\projz(\cminus)$ both coincide with the $z$-axis. If one restricts to $b$ functions independent of $z$, we have

\begin{equation}
    1 = b_1 = b_0 = b_2 < \bM = \frac{6}{5},
    \label{eq:mn_const_b_hierarchy}
\end{equation}

\noindent
independent of $s$: in particular, this holds for the spherical case (the Plummer sphere).

\begin{figure*}
\centering
\includegraphics[width=\linewidth]{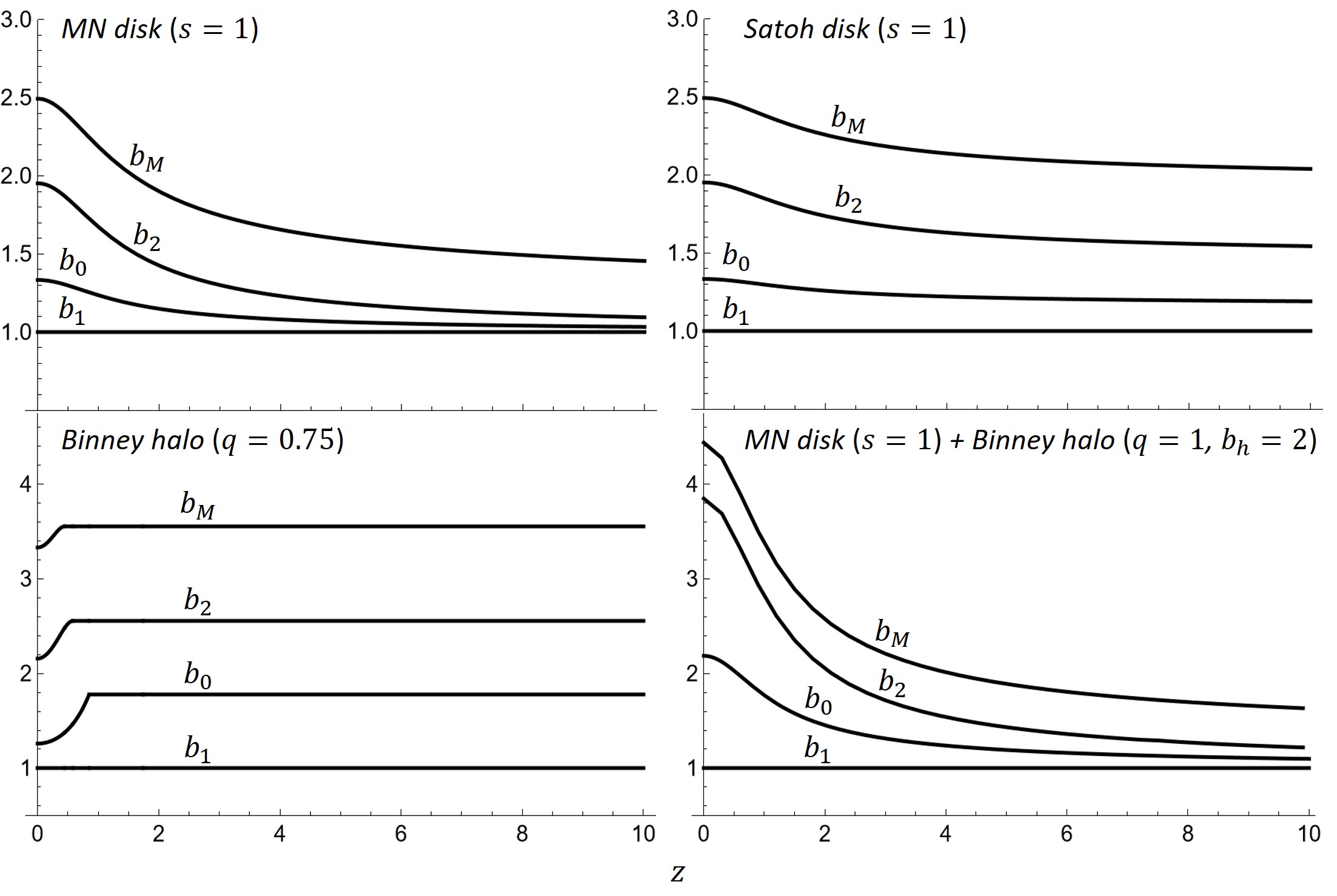}
\caption{Profiles of $b_1(z)$, $b_0(z)$, $b_2(z)$, $\bM(z)$ for the MN disk with $s=1$ (top-left panel), the Satoh disk with $s=1$ (top-right panel), the Binney logarithmic halo with $q=0.75$ (bottom-left panel), and the two-component MN disk with $s=1$ in a dominant ($\Phi = \phih$), spherical ($q=1$) Binney logarithmic halo with $\bh = 2$ (bottom-right panel). The $z$ values are normalized to $\bh$ for the Binney model, and to $b_*$ for the other models. Because for these models $\left[\rho, \Phi\right] \ge 0$ everywhere, it follows that $b_1(z) \le b_2(z) \le b_0(z) \le \bM(z)$ for all $z$, as discussed in Section \ref{sec:phys_cons_bans}.}
\label{fig:1comp_bprof}
\end{figure*}

Figure \ref{fig:MN_b_Panel} illustrates the behaviour of $\vphiqm$, $\Deltac$ and $\Delta$ for the MN disk with $s = 1$, and for three values of $b$, i.e. $b = 0.5, 1, 2$. The adopted values map three different kinematical configurations, as can be seen from equation (\ref{eq:mn_const_b_hierarchy}). For $b = 0.5$, $\vphiqm$ and $\Deltac$ are everywhere positive, while $\Delta$ is expected to be negative over some region in $\bplus \cap \cminus = \cminus$, being $0.5 < b_1$, and positive over $\bminus \cap \cplus = \bminus$, being $0.5 < b_2$. This is nicely confirmed by the three left panels in Figure \ref{fig:MN_b_Panel}, where white regions indicate negative values of the fields. The three central panels correspond to $b = 1$, i.e. to the two-integral solutions (see also the three top panels in Figure \ref{fig:2int_1comp_panel}). The fields in this case are everywhere positive, in agreement with the general discussion and equation (\ref{eq:mn_const_b_hierarchy}). Note how the increase of $b$ from $0.5$ to $1$ produces the expected changes summarized in Figure \ref{fig:fields_table}, i.e. $\vphiqm$ decreases over $\bminus$ and increases over $\bplus$, $\Deltac$ decreases everywhere being $\dminus$ for this model coincident with the ($R,z$)-plane, and $\Delta$ behaves qualitatively as $\vphiqm$, becoming non-negative over $\cminus$ because $1 = b_1$. Finally, the three right panels of Figure \ref{fig:MN_b_Panel} correspond to the unphysical model with $b > \bM$, and a large region in the $\vphiqm$ map becomes white: increasing further $b$ would increase the size of the white region up to the whole $\bminus$. However, $\vphiqm$ continues to increase in $\bplus$. As $2 > b_0$, $\Deltac$ is now negative over a large portion of the ($R,z$)-plane and, increasing further $b$, $\Deltac$ would become negative everywhere. Finally, the negative region of $\Delta$ has now switched to the $\bminus$ region, being $2 > b_2$. Interestingly, the shape and the position of the white region of $\vphiqm$ resemble those where 
the numerically evaluated $\vphiqm$ for the similar model of the Satoh disk with $b=4/3$ (\citealt{Cappellari2020MNRAS.494.4819C}, Figure 10) goes to zero.

\begin{figure*}
\centering
        \includegraphics[width=\linewidth]{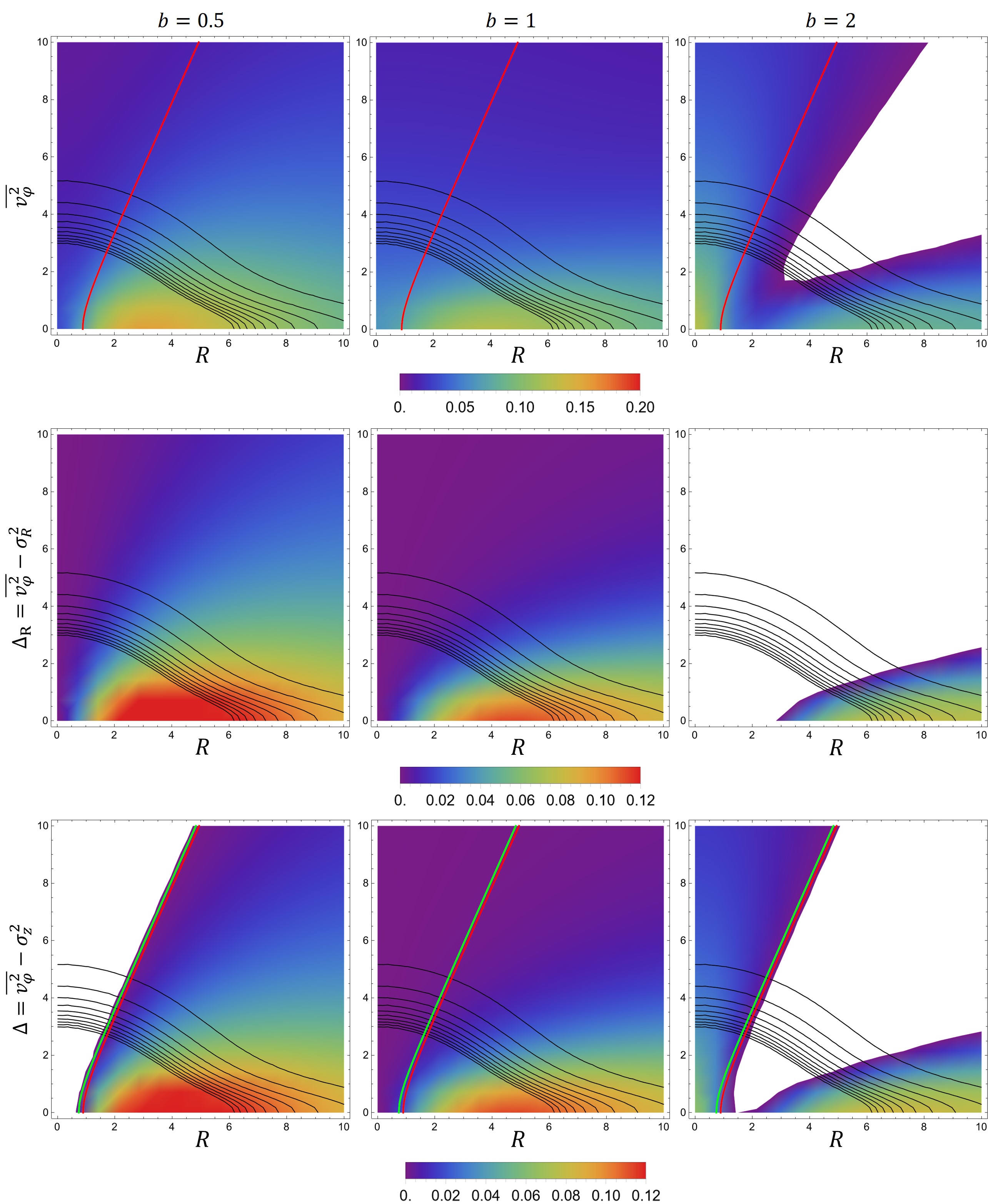}
        \caption{Maps in the ($R,z$)-plane of $\vphiqm$ (top row), $\Deltac$ (central row) and $\Delta$ (bottom row) in units of $GM_*/b_*$ for the MN disk with $s=1$ and constant $b=0.5$ (left column), $b=1$ (central column) and $b=2$ (right column). Black solid lines show equally spaced isodensity contours of the stellar distribution. The red and green lines show, respectively, the $\bzero$ and $\czero$ curves, while $\dzero$ coincides with the $z$-axis. White regions correspond to negative values of the fields.}
\label{fig:MN_b_Panel}
\end{figure*}

\subsection{The Satoh disk}

The potential-density pair of the \citet{1980PASJ...32...41S} disk of total mass $M_*$, and scale lengths $a_*$ and $b_*$, can be written as:

\begin{equation}
\begin{dcases}
    \phis = -\frac{GM_*}{b_* \xi}, \quad \xi \equiv \sqrt{R^2 + (s+\zeta)^2 - 1},\\ 
     \rhos = \frac{M_*}{b_*^3}\frac{s\xi^2 + 3s\zeta(s+2\zeta)}{4 \pi \zeta^3 \xi^5}, 
\end{dcases}
\label{eq:S_dp_pair}
\end{equation}

\noindent
where again $\zeta \equiv \sqrt{1+z^2}$, $s \equiv a_*/b_*$ measures the flattening of the disk, and all lengths are assumed to be normalized to $b_* \ne 0$ as throughout the Section. Notice that, at variance with the MN model, the spherical limit of the Satoh model ($s = 0$) reduces to the point mass case of no practical interest. From equations (\ref{eq:2int_JE_sol}) and (\ref{eq:jeans_rad_comm}), one has:

\begin{equation}
    \rhos\sigma_z^2 = \frac{G M_*^2}{b_*^4} \frac{s(s+2\zeta)}{8\pi \zeta^2 \xi^6},\quad \left[\rho_*,\Phi_*\right] = \frac{G M_*^2}{b_*^5}\frac{sR}{4\pi \zeta^3 \xi^6}.
\label{eq:2int_S_sol}
\end{equation}

\noindent
Similarly to the MN disk, $\left[\rho_*,\Phi_*\right] \ge 0$ everywhere, and also the maps of $\sigma_z^2$, $\left[\rhos,\Phi\right]$ and $R \partial\Phi/\partial R$ are very similar to those of the MN model in Figure \ref{fig:2int_1comp_panel}, and thus are not shown.

The $\mathscr{B}$, $\mathscr{C}$ and $\mathscr{D}$ regions can be determined analytically. In particular, we have

\begin{equation}
    \bminus = \left\{(R,z) : 5R^2 - (s + \zeta)^2 + 1 > 0 \right\},
    \label{eq:s_b_exp}
\end{equation}

\noindent
so that $\bminus$ is the portion of the ($R,z$)-plane to the right of the curve $\bzero$ with vertex coordinates $z = 0$ and $R = \sqrt{s(s+2)/5}$, and asymptotes $z =\nobreak \pm\sqrt{5} R$; $\projz(\bminus)$ coincides with the $z$-axis. The radial derivative of the first of equation (\ref{eq:2int_S_sol}) shows that $D < 0$ everywhere, except for the $z$-axis, where $D = 0$; therefore, as for the MN disk, $\projz(\dminus)$ coincides with the $z$-axis, and the set $\dplus$ is empty. Finally, from the second of equation (\ref{eq:delta_bans}),

\begin{equation}
    \begin{split}
        \cminus = \Big\{(R,z) :\, &2R^4 + R^2 (12 \zeta^2 + 9 s \zeta + 2 s^2 - 2)\\
        - &\left(\zeta+s+1\right)(\zeta+s-1) (s + 2\zeta) \zeta < 0 \Big\},
    \end{split}
\label{eq:cmzero_S}
\end{equation}

\noindent
and again, similarly to the MN disk, the resulting expression is a biquadratic with two real solutions. From the fact that $\zeta \ge 1$, the coefficients present a permanence and a variation of their sign: therefore, $\cminus$ is the region between the $z$-axis and the (positive) square root of the (positive) solution of the biquadratic. The explicit expression of $\czero$ can be obtained without difficulty; we just notice that, as for the MN disk, $\projz(\cminus)$ coincides with the $z$-axis. Again, all the results summarized in Section \ref{sec:phys_cons_bans} about the relative position of the regions apply, i.e. $\bminus \subseteq \dminus$, and $\bminus \cap \cplus = \bminus$, $\bplus \cap \cminus = \cminus$.

The limitations on $b(z)$ for the positivity of $\vphiqm$, $\Deltac$ and $\Delta$ can all be obtained analytically. From equations (\ref{eq:bmz_criterion}), (\ref{eq:b0z_criterion}), (\ref{eq:b1_def}) and (\ref{eq:b2_def}), 

\begin{equation}
\begin{split}
   \bM(z) = \frac{6}{5} + \frac{14\zeta_*^2  + 2\zeta_* \sqrt{24\zeta_*^2 + 60\zeta(s+2\zeta)}}{25\zeta(s+2\zeta)},
\end{split}
\end{equation}

\begin{equation}
    b_0(z) = 1 + \frac{\zeta_*^2}{3\zeta(s+2\zeta)},
    \label{eq:b0_S}
\end{equation}

\begin{equation}
    b_1(z) = 1,\quad b_2(z) = 1 + \frac{(14 + 4\sqrt{6})\zeta_*^2}{25\zeta(s+2\zeta)},
\end{equation}

\noindent
where $\zeta_* \equiv \sqrt{(s + \zeta)^2 - 1}$. In the top-right panel of Figure \ref{fig:1comp_bprof}, we show the four functions above for $s = 1$; again, they fulfil the inequalities in equation (\ref{eq:bz_hierarchy}), considering that $\projz(\bminus)$ and $\projz(\cminus)$ both coincide with the $z$-axis. For a spatially constant $b$, we have

\begin{equation}
    b_1 = 1 < b_0 = \frac{7}{6} < b_2 = \frac{32 + 2\sqrt{6}}{25} < \bM = \frac{49}{25},
\end{equation}

\noindent
independently of $s$.

We do not show the analogous of Figure \ref{fig:MN_b_Panel}, due to the close similarity, for sufficiently flat models, with the MN case, so that all comments made for the MN disk apply. The only noticeable difference is that now $b_1 \neq b_0 \neq b_2$, while for the MN disk they are the same [see equation (\ref{eq:mn_const_b_hierarchy})].

\subsection{The Binney logarithmic halo}\label{sec:B_model}

The potential-density pair of the Binney logarithmic halo (BT08) of asymptotic circular \nobreak velocity $\vh$, scale length $\bh$, and potential flattening $q$ is:

\begin{equation}
\begin{dcases}
    \phih = \frac{\vh^2}{2}\ln\left(1 + R^2 + \frac{z^2}{q^2}\right),\\ 
    \rhoh = \frac{\vh^2}{4\pi G \bh^2}\frac{1 + 2q^2 + R^2 + \left(2-q^{-2}\right)z^2}{q^2 \left(1 + R^2 + q^{-2}\,z^2\right)^2},
\end{dcases}
\label{eq:B_pd_pair}
\end{equation}



\noindent
where throughout this Section the lengths are assumed normalized to $\bh \ne 0$. As well known, for $q<1/\sqrt{2}$, $\rhoh$ becomes negative on the $z$-axis. The dynamical properties of this model are given by \citet{Evans1993MNRAS.260..191E} (see also \citealt{Evans1994MNRAS.267..333E} and  \citealt{EvansdZ1994MNRAS.271..202E} for the properties of the larger family of the so-called power-law models); here we just use equations (\ref{eq:2int_JE_sol}) and (\ref{eq:jeans_rad_comm}) to obtain

\begin{equation}
\begin{dcases}
    \rhoh\sigma_z^2 = \frac{\vh^4}{G \bh^2} \frac{q^4R^2 + (2q^2 - 1)z^2 + 2q^4}{8\pi (q^2R^2 + z^2 + q^2)^2},\\
    \left[\rhoh,\phih\right] = \frac{\vh^4}{G \bh^3} \frac{q^2 (1 - q^2) R}{4\pi (q^2 R^2 + z^2 + q^2)^2},
\end{dcases}
    \label{eq:2int_B_sol}
\end{equation}

\noindent
so that $\sigma_z^2 \ge 0$ for $q \ge 1/\sqrt{2}$. Moreover, $\left[\rhoh,\phih\right] \ge 0$ for $1/\sqrt{2} \le q < 1$, it vanishes for $q = 1$ (the spherical limit), and it is $\le 0$ for $q>1$. We finally notice that, in the two-integral case,

\begin{equation}
    \rhoh\vphiqm = \frac{\vh^4}{G \bh^2}\frac{q^2 (2 - q^2) R^2 + (2q^2 - 1) z^2 - 2q^4}{8\pi (q^2 R^2 + z^2 + q^2)^2},
    \label{eq:vphiqm_B_exp1}
\end{equation}

\noindent
showing that $\vphiqm$ becomes negative for $q>\sqrt{2}$, and the models are hence inconsistent. Actually, \citet{Evans1993MNRAS.260..191E} proved that the two-integral DF becomes negative for $q \gtrsim 1.08$ so that models with $1.08 \lesssim q \le \sqrt{2}$ do not exist\footnote{This case exemplifies that $\sigma_z^2 \ge 0$ and $\vphiqm \ge 0$ are only \textit{necessary} conditions for the model consistency, while $\sigma_z^2 < 0$ or $\vphiqm < 0$ are \textit{sufficient} to prove the model inconsistency.} although they have non-negative $\sigma_z^2$ and $\vphiqm$. From now on, we consider models with $1/\sqrt{2} \le q \le 1$, so that $\left[\rhoh,\phih\right] \ge 0$ everywhere: in the central panels of Figure \ref{fig:2int_1comp_panel}, the 2D maps of $\sigma_z^2$, $\left[\rhoh,\phih\right]$, and $R \partial\Phi/\partial R$ are shown for $q=0.75$.

As for the two previous models, the $\mathscr{B}$, $\mathscr{C}$ and $\mathscr{D}$ regions can be determined in a fully analytical way. We notice however that for $q \simeq 1$ the Binney halo is qualitatively different from a disk, resembling a spheroid; therefore, we expect some important differences compared to the previously discussed cases. From equation (\ref{eq:2int_B_sol}), we have

\begin{equation}
\begin{split}
    \bminus = \Big\{(R,z) :\, &q^6 R^4 + 3 q^2 (q^2 z^2 - z^2 + q^4) R^2 -\\
    &(q^2 + z^2) (2 q^2 z^2 - z^2 + 2 q ^4) > 0 \Big\}.
 \end{split}
 \label{eq:B_b_exp}
\end{equation}

\noindent
With some work, it can be proved that the discriminant of the biquadratic is strictly positive independently of $q$; moreover, for $q \ge 1/\sqrt{2}$, the last term in equation (\ref{eq:B_b_exp}) is negative, so that a permanence and a variation of the sign of its coefficients are present. We conclude that $\bzero$ is the (positive) square root of the (positive) solution, and $\bminus$ is the region of the ($R,z$)-plane to the right of $\bzero$. As for the previous models, $\projz(\bminus)$ coincides with the $z$-axis; for simplicity, we do not report the explicit expression of $\bzero$ here, but we show it as the red line in the bottom panels of Figure \ref{fig:MN_B_Reg}, for three representative values of $q$. The properties of the $\mathscr{D}$ region are more complicated compared to the previous cases: in fact, while for $q \ge \sqrt{2/3}$, $D < 0$ everywhere (vanishing on the $z$-axis) and thus $\dminus$ fully covers the ($R,z$)-plane, for $q < \sqrt{2/3}$

\begin{equation}
\begin{split}
    \dminus = \Bigg\{(R,z) : z - \frac{q^2 \sqrt{3 + R^2}}{\sqrt{2-3q^2}} < 0\Bigg\},
\end{split}
\end{equation}

\noindent
and the $\dplus$ region now exists: for example, in the bottom right panel of Figure \ref{fig:MN_B_Reg}, the dashed line shows $\dzero$ for $q = 0.75$, and $\dplus$ is the part of the plane above it. Notice how, independently of $q$, $\projz(\dminus)$ coincides with the $z$-axis. Finally, from equation (\ref{eq:delta_bans}),

\begin{equation}
    \begin{split}
        \cminus = \Big\{&(R,z) : q^4 (2- q^2) R^4 + q^2 (q^2 z^2 - z^2 + q^4 +\\ &2q^2) R^2
        - (q^2 + z^2) (2 q^2 z^2 - z^2 + 2 q^4) < 0 \Big\},
    \end{split}
\end{equation}

\noindent
where, in the considered range of $1/\sqrt{2} \le q \le 1$, the discriminant of the biquadratic is positive, and the coefficients present a permanence and a variation of the sign. Therefore, $\czero$ is the (positive) square root of the (positive) solution, and $\cminus$ is the region between the $z$-axis and $\czero$; $\projz(\cminus)$ coincides with the $z$-axis. In the bottom panels of Figure \ref{fig:MN_B_Reg}, the green line shows $\czero$ for three representative values of $q$. It is also apparent how $\bminus \subseteq \dminus$, $\bminus \cap \cplus = \bminus$ and $\bplus \cap \cminus = \cminus$, as expected from the general discussion; similarly to the MN model, notice that the $\bplus \cap \cplus$ region becomes larger as $q$ decreases, i.e. for more flattened systems. Again, in the spherical case ($q = 1$) we would obtain $\bzero = \czero$, $\bminus = \cplus$ and $\bplus = \cminus$.

We now determine the constraints on $b(z)$ from the request of positivity for $\vphiqm$, $\Deltac$ and $\Delta$. From equation (\ref{eq:bmz_criterion}), with $1/\sqrt{2} \le q \le 1$,

\begin{equation}
        \bM(z) = \frac{2}{q^2},\quad z \ge z_{\mathrm{M}}(q) \equiv \frac{q^2\sqrt{1-q^2}}{\sqrt{q^4 - 2q^2 + 3/2}},
        \label{eq:binney_zcM}
\end{equation}

\noindent
while, for $z \le z_{\mathrm{M}}$, $\bM(z)$ is a complicated function monotonically decreasing from $2/q^2$ to the minimum

\begin{equation}
    \bM(0) = \frac{14 + 12 q^2 + 8\sqrt{2+q^2-2q^4}}{17 q^2},
    \label{eq:bM_B_exp}
\end{equation}

\noindent
reached on the equatorial plane; in the spherical case ($q=1$), $\bM(z) = 2$ for all $z$. From equation ($\ref{eq:b0z_criterion}$),

\begin{equation}
    b_0(z) = \frac{1}{q^2},\quad z \ge z_{0}(q) \equiv \frac{q^2}{\sqrt{1-q^2}},
    \label{eq:binney_zc0}
\end{equation}

\noindent
and, for $z \le z_{0}$, $b_0(z)$ is a complicated function monotonically decreasing from $1/q^2$ to the minimum

\begin{equation}
    b_0(0) = \frac{2}{3}+\frac{1}{3q^2},
    \label{eq:B_b0C_exp}
\end{equation}

\noindent
reached again on the equatorial plane; in the spherical case, $b_0(z) = 1$ for all $z$, as expected from the general discussion. Finally, from equations (\ref{eq:b1_def}) and (\ref{eq:b2_def}), $b_1(z) = 1$, and

\begin{equation}
    b_2(z) = \frac{2}{q^2} - 1,\quad z \ge z_{2}(q) \equiv \frac{\sqrt{2}q^2}{\sqrt{3 - 2q^2}},
    \label{eq:binney_bmp_zc}
\end{equation}

\noindent
while, for $z \le z_2$, $b_2(z)$ is again a complicated function which decreases monotonically, reaching its minimum

\begin{equation}
    b_2(0) = \frac{14 + 8\sqrt{2} + (3-8\sqrt{2})q^2}{17q^2},
    \label{eq:B_b2C_exp}
\end{equation}

\noindent
on the equatorial plane; in the spherical case, $b_2(z) = 1$. In the bottom-left panel of Figure \ref{fig:1comp_bprof}, we show the four $b(z)$ functions described above for the representative case of $q = 0.75$: the chain of inequalities in equation (\ref{eq:bz_hierarchy}) is apparent, even though the profiles qualitatively differ from those of the two disk cases. In particular, if one restricts to the constant $b$ case, we have $b_1 = 1 \le b_0(0) \le b_2(0) < \bM(0)$.

Figure \ref{fig:Binney_b_panel} shows the 2D maps of $\vphiqm$, $\Deltac$ and $\Delta$ for the Binney logarithmic halo with $q = 0.75$ and for three values of $b$, i.e. $b = 0.5, 1, 3.75$. The adopted values of $b$ allow to illustrate three representative kinematical behaviours. When $b = 0.5$, we expect $\vphiqm$ and $\Deltac$ to be everywhere positive, while a negative $\Delta$ is expected over some region in $\cminus$, being $0.5 < b_1 = 1$, and positive over $\bminus$, being $0.5 < b_2 \simeq 2.16$. This is nicely confirmed by the three left panels, where white regions indicate negative values. The three central panels correspond to $b = 1$, i.e. to the two-integral solutions, and complement the three middle panels in Figure \ref{fig:2int_1comp_panel}. The fields for $b = 1$ case are everywhere positive, in agreement with the general discussion and equations (\ref{eq:bM_B_exp}), (\ref{eq:B_b0C_exp}) and (\ref{eq:B_b2C_exp}). The increase of $b$ from $0.5$ to $1$ produces, as expected, that $\vphiqm$ decreases over $\bminus$ and increases over $\bplus$, $\Deltac$ decreases over $\dminus$ and increases over $\dplus$, and $\Delta$ behaves qualitatively as $\vphiqm$, becoming non-negative over $\cminus$ because $1 = b_1$. Finally, the three right panels correspond to the unphysical model with $b > \bM \simeq 3.33$, and a portion of $\bminus$ is already white; compared to the disk models, this region is more confined near the equatorial plane. Increasing further $b$ would extend the white region to cover the whole $\bminus$, while $\vphiqm$ would continue increasing in $\bplus$. As $3.75 > b_0 \simeq 1.26$, $\Deltac$ is now negative over a large portion of $\dminus$ and, increasing further $b$, $\Deltac$ would become negative over the whole $\dminus$ region, i.e. below the yellow line. At variance with disks, however, $\Deltac$ would remain positive and increase above the line. Finally, the negative region of $\Delta$ has now switched to the $\bminus$ region, being $3.75 > b_2$.

\begin{figure*}
\centering
        \includegraphics[width=\linewidth]{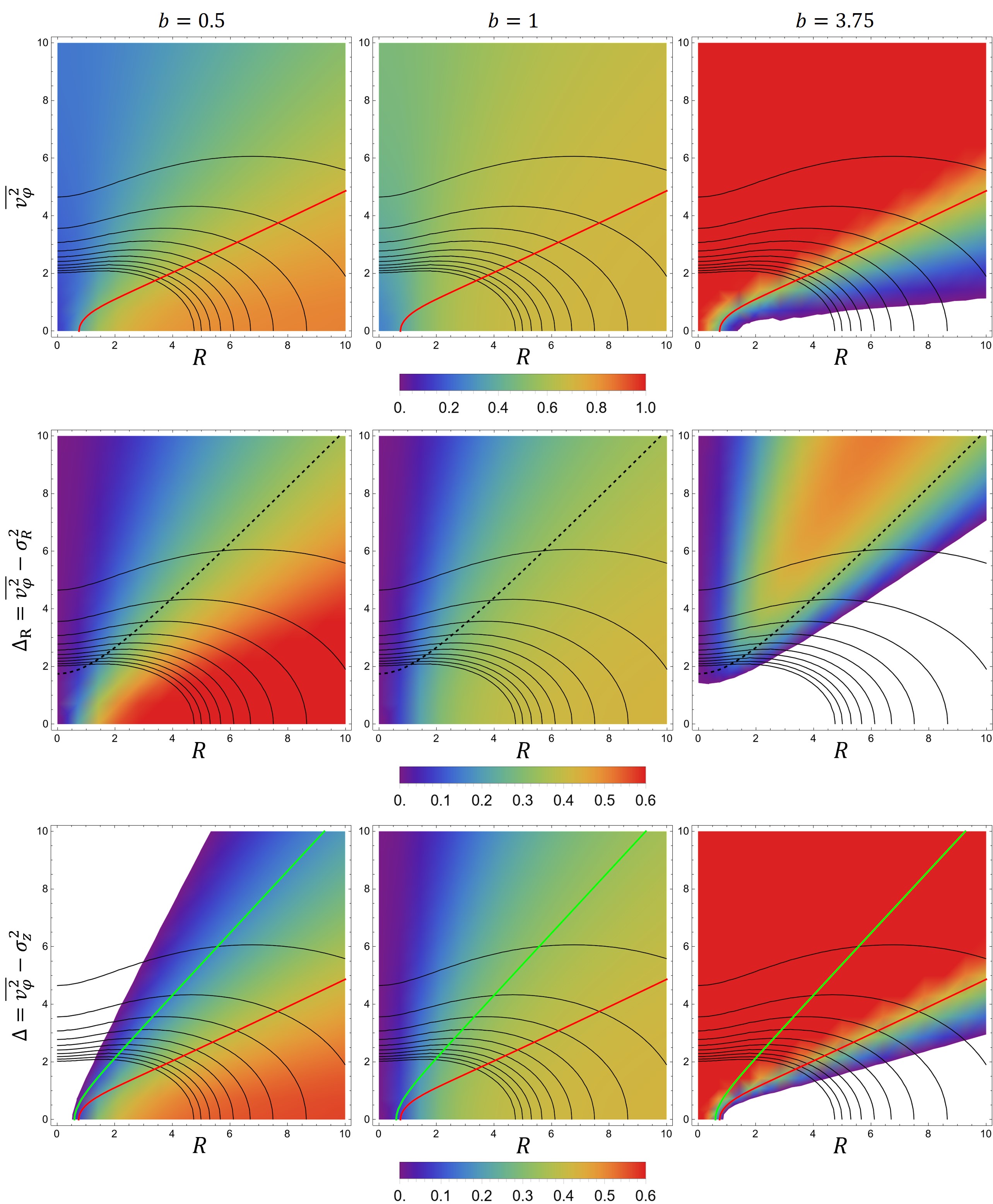}
        \caption{Maps in the ($R,z$)-plane of $\vphiqm$ (top row), $\Deltac$ (central row) and $\Delta$ (bottom row) in units of $\vh^2$, for the Binney logarithmic halo with $q = 0.75$ and constant $b=0.5$ (left column), $b=1$ (central column) and $b=2$ (right column). Black solid lines show equally spaced isodensity contours of the stellar distribution. The red, green and black dashed lines show, respectively, the $\bzero$, $\czero$ and $\dzero$ curves. White regions correspond to negative values of the fields.}
\label{fig:Binney_b_panel}
\end{figure*}

\section{A two-component model: the Miyamoto-Nagai disk in a Binney logarithmic halo}\label{sec:mnb_model}

We now study whether and how the presence of a dark matter (DM) halo changes the conclusions obtained for one-component models. In particular, we consider the two-component model made by a stellar MN disk embedded in a DM Binney logarithmic halo. We adopt this model because $\sigma_z^2$ and the commutator for the stellar distribution can be given in closed form for arbitrary values of $s$ and $q$ (\citealt{2015MNRAS.448.2921S}, hereafter S15); the resulting expressions are however sufficiently complicated to exclude the possibility of a simple analytical study of the $\mathscr{B}$, $\mathscr{C}$ and $\mathscr{D}$ regions, which are determined in a numerical way starting from the analytical formulae in S15.

By using the normalizations adopted for the MN model, i.e. $b_*$ for lengths and $GM_*/b_*$ for potentials, the dimensionless total potential is

\begin{equation}
    \Phi = \phis + \mathscr{R}\phih,\quad \mathscr{R} \equiv \frac{\vh^2 b_*}{GM_*},
    \label{eq:MNB_phitot}
\end{equation}

\noindent
in which $\phis = -1/\xi$, as given in equation (\ref{eq:MN_pd_pair}), and from equation (\ref{eq:B_pd_pair}) the dimensionless halo potential is

\begin{equation}
    \phih = \frac{1}{2} \ln\left( 1 + \frac{R^2}{\bh^2} + \frac{z^2}{q^2 \bh^2}\right),
\end{equation}

\noindent
where $\bh$ is the scale length of the halo now in units of $b_*$, i.e. a dimensionless parameter. For $\mathscr{R}=0$, the model reduces to the self-gravitating MN disk, while for $\mathscr{R} \gg 1$ we obtain the \lq\lq halo-dominated" case, in practice equivalent to considering $\Phi = \phih$ in the JEs. Notice that, at sufficiently large $r = \sqrt{R^2 + z^2}$, the model is always halo-dominated, independently of the value of $\mathscr{R}$.

In the two-integral case, $\sigma_z^2$ of the stellar distribution can be obtained by adding equations (10) and (17) of S15, while $\left[\rhos,\Phi\right]$ can be obtained from equation (27) of S15. It can be shown that, while $\sigma_z^2 \ge 0$ regardless of the model parameters, the sign of $\left[\rhos,\Phi\right]$ depends on the value of $q$, and if $q < 1$, there are regions at large $r$ where $\left[\rhos,\Phi\right] < 0$. Instead, if $q \ge 1$, then $\left[\rhos,\Phi\right] \ge 0$ everywhere. We notice that an asymptotic analysis of the solutions of the JEs shows that, independently of $s$ and $\mathscr{R}$, $\vphiqm$ becomes negative at large distances from the origin for $q \lesssim 0.85$; instead, for $q \gtrsim 0.85$, $\vphiqm \ge 0$ everywhere.

As we expect that the halo-dominated models are the most different compared to the one-component MN disk, in the following we focus on these models, and for simplicity, we consider a spherical $(q = 1)$ halo. In the three bottom panels of Figure \ref{fig:2int_1comp_panel}, we show the maps of $\sigma_z^2$, $\left[\rhos,\Phi\right]$ and $R\partial\Phi/\partial R$ for a halo-dominated MN disk with $s=1$, and $\bh = 2$. By comparison with the MN panels in the figure, it is apparent that the main effect of the halo is to produce an increase of $\sigma_z^2$, and of the commutator at large radii near the equatorial plane, as expected given the flat rotation curve of the halo at large radii. Instead, in the more central regions, the self-gravitating effects of the disk are not visible because the gravitational field of the disk is neglected everywhere in the presented model.

However, even though the gravitational potential is dominated by the halo, the disk density distribution enters the $B$, $C$ and $D$ functions, so we expect that the $\mathscr{B}$, $\mathscr{C}$ and $\mathscr{D}$ regions, whose determination now must be carried out numerically, are not coincident with those of the Binney halo (see Section \ref{sec:B_model}). Indeed, now the $\mathscr{B}$ and $\mathscr{C}$ regions are more similar to those of the MN disk, and $\dminus$ coincides with the ($R,z$)-plane, with $D = 0$ on the $z$-axis.

Due to the numerical determination of the $\mathscr{B}$, $\mathscr{C}$ and $\mathscr{D}$ regions, also the various constraints on the $b$ values had to be found numerically. In the bottom-right panel of Figure \ref{fig:1comp_bprof}, we show the four $b(z)$ constraints for the halo-dominated model with $q = 1$, $s = 1$ and $\bh = 2$. Reassuringly, they obey the inequality in equation (\ref{eq:bz_hierarchy}), being the commutator positive everywhere. For the constant $b$ case, they can be obtained analytically through an asymptotic analysis, and we found that $b_1 = b_0 = b_2 = 1$, and $\bM = 5/4$.

Figure \ref{fig:MNB_bPanel_bh2} shows the 2D maps of $\vphiqm$, $\Deltac$ and $\Delta$ for the same $q = 1$, $s = 1$ and $\bh = 2$, and for $b = 0.5, 1, 2$ (as in Figure \ref{fig:MN_b_Panel} for the MN disk). The main result is that the fields are more similar to those of the one-component MN disk in Figure \ref{fig:MN_b_Panel} than to those of the one-component Binney model in Figure \ref{fig:Binney_b_panel}, even though the MN disk is halo-dominated. Therefore, we conclude that under the $b$-ansatz a DM halo does not alter appreciably the qualitative behaviour of the kinematical fields of the disk, even when the halo is very massive.

\begin{figure*}
\centering
        \includegraphics[width=\linewidth]{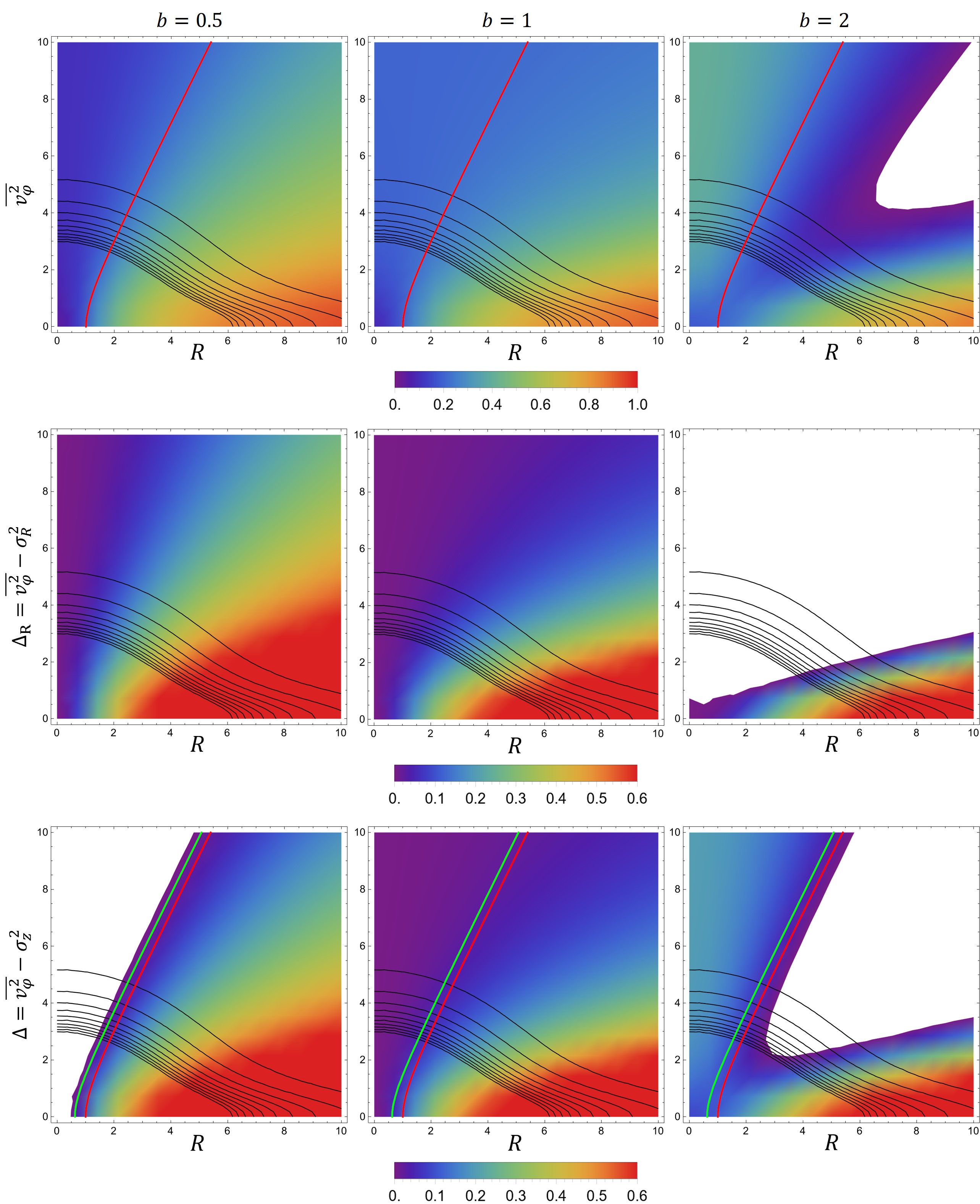}
        \caption{Maps in the ($R,z$)-plane of $\vphiqm$ (top row), $\Deltac$ (central row) and $\Delta$ (bottom row) in units of $GM_*/b_*$ for the MN disk with $s=1$ in a dominant ($\Phi = \phih$), spherical ($q=1$) Binney logarithmic halo with $\bh = 2$, and constant $b=0.5$ (left column), $b=1$ (central column) and $b=2$ (right column). Black solid lines show equally spaced isodensity contours of the stellar distribution. The red and green lines show, respectively, the $\bzero$ and $\czero$ curves. White regions correspond to negative values of the fields.}
\label{fig:MNB_bPanel_bh2}
\end{figure*}

\section{Summary and conclusions}\label{sec:disc_conc}

Unless a specific functional form of the phase-space DF is assumed, the JEs do not form a closed system of equations, and to solve them a choice must be made for the link between the velocity moments. In this work, we considered the JEs of axisymmetric systems with three unknowns ($\sigma_z^2$, $\sigma_R^2$ and $\vphiqm$), and focused on the choice of $\sigma_R^2=b(R,z)\sigma_z^2$ (the \lq\lq b-ansatz", where $b$ is a positive function), introduced in \citet{Cappellari2008MNRAS.390...71C} and widely used in the dynamical modelling of stellar systems. After recalling the JEs in cylindrical coordinates, and the expressions for the ansatz-independent quantities $\sigma_z^2$ and $[\rhos,\Phi]$, and the associated functions $B$, $C$ and $D$ defined over the ($R,z$)-plane, we investigated the behaviour of the kinematical quantities that depend on $b$ as follows:

\begin{itemize}
    \item We first gave the general expressions, as a function of $b$ and of the ansatz-independent quantities, for $\vphiqm$ and the related functions $\Delta = \vphiqm - \sigma_z^2$ and $\Deltac = \vphiqm - \sigma_R^2$, which enter the \citet{1980PASJ...32...41S} decomposition of $\vphiqm$ to derive the ordered rotational velocity $\vphim$. These general expressions are fundamental to investigate the positivity of $\vphiqm$, required for the model consistency, and of $\Delta$ and $\Deltac$, necessary for the \citet{1980PASJ...32...41S} decomposition and its variants. These general expressions also allow to predict, before solving the JEs, the trends of $\vphiqm$, $\Delta$ and $\Deltac$ as a function of $b$ throughout the galaxy, thus avoiding a time-consuming numerical exploration of the parameter space in the modelling.
    \item For systems with $\partial \Phi/\partial R \ge 0$ everywhere (the common situation), and $b = b(z)$, we discussed the positivity problem in full generality: we first showed how to determine the ansatz-independent regions of the $(R,z)$-plane where negativity is possible, i.e. $\bminus$ for $\vphiqm$, $\dminus$ for $\Deltac$ and a more complicate intersection of $\bpm$ and $\cpm$ for $\Delta$. Then, we showed how to determine the limiting values of $b(z)$ that guarantee positive kinematical fields over these regions: $b(z) \le \bM(z)$ for $\vphiqm \ge 0$, $b(z) \le b_0(z)$ for $\Deltac \ge 0$, and $b_1(z) \le b(z) \le b_2(z)$ for $\Delta \geq 0$. In the common case of $[\rhos,\Phi] \ge 0$ everywhere (as for oblate self-gravitating systems), one has $b_1(z) \le 1$ for $z \in \projz(\cminus)$, and $1 \le b_0(z) \le b_2(z) \le \bM(z)$ for $z \in \projz(\bminus)$.
    \item In particular, for systems with $[\rhos,\Phi] \ge 0$, we showed that as $b$ increases, in addition to the trivial fact that $\sigma_R^2 = b \sigma_z^2$ increases, $\vphiqm$ and $\Delta$ decrease over $\bminus$ and increase over $\bplus$, and $\Deltac$ decreases over $\dminus$ and increases over $\dplus$.
\end{itemize}

After the general analysis of the problem and the setup of the procedure to investigate the kinematical properties of a model as a function of $b$, we illustrated the method with three widely used galaxy models: the Miyamoto \& Nagai and Satoh disks, and the Binney logarithmic halo. In doing so, we obtained the following results:

\begin{itemize}
    \item The shape of the $\bpm$, $\cpm$ and $\dpm$ regions and the limits on $b(z)$ can all be obtained analytically. For the two disks, $\dminus$ (the region where $\partial \rhos\sigma_z^2 / \partial R < 0$) coincides with the ($R,z$)-plane, independently of the disk flattening; for the Binney halo, instead, there is a critical value of the potential flattening such that for flatter potentials the region $\dplus$ appears, bounded by a curve containing the $z$-axis. For all the models, $\bminus$ and $\cminus$ are also bounded by hyperbola-like curves but contain the $R$-axis. Concerning the limits on $b(z)$, for all the models $\projz(\bminus) = \projz(\cminus) = \projz(\dminus)$ coincide with the $z$-axis, and $b_1(z) \le 1 \le b_0(z) \le b_2(z) \le \bM(z)$.
    \item With the aid of two-dimensional maps, we illustrated how the kinematical fields change over the galaxy by changing $b$, assumed spatially constant. The results are in agreement with the general considerations, both for their trend with $b$, and for the regions where they become negative, confirming the utility of the preliminary analysis.
    \item We finally applied the procedure to the two-component model made by a stellar MN disk embedded in a dominant, spherical dark matter Binney halo; in this case, the investigation has been carried out numerically. The dark halo does not alter appreciably the qualitative trend of the kinematical fields of the stellar disk as a function of $b$, even when the halo is very massive.
\end{itemize}


  A final comment is worth on the relation between
  $\sigma_{\varphi}^2$ and $\sigma_R^2$ as a function of $b$ and
  $\gamma$, the other parameter that measures the orbital anisotropy,
  with positive values corresponding to radial anisotropy, and
  negative ones to tangential anisotropy [see equations (\ref{eq:gamma_def}) and (\ref{eq:gamma_cons})]. As noticed in Section \ref{sec:b_ans}, $\gamma$ is necessarily positive
  where $\Deltac<0$; where $\Deltac >0$, instead,
  $\gamma$ can be positive ($k^2>1$ in the Satoh decomposition) or
  negative ($k^2<1$). Allowing for $b$ and $\gamma$ as free
  parameters in the JAM method,  the stellar kinematics of regular rotators in the ATLAS$^{\rm 3D}$ survey could be successfully modelled (e.g., \citealt{Cappellari2016ARA&A..54..597C}).  It turned out that on average $b>1$ and $\gamma \approx 0$, within about one effective radius; a limit was also found in the orbital anisotropy, of $\beta =1-1/b \la 0.7 \epsilon_{\mathrm{intr}}$, where $\epsilon_{\mathrm{intr}}$ is the intrinsic flattening of the galaxy. This limit had first been found with the Schwarzschild orbit superposition method for the galaxies of the SAURON survey
  (\citealt{Cappellari2007MNRAS.379..418C}), and thereafter has been often shown in the $(V/\sigma,\epsilon)$, and $(\lambda_R,\epsilon)$ diagrams of larger samples of regular rotators, until those of the MaNGA survey (e.g., \citealt{Wang2020MNRAS.495.1958W}). In a recent work, \citet{Wang2021MNRAS.500L..27W}  investigated the origin of this limit: with the JAM modeling, they built mock samples of axisymmetric galaxies with a velocity dispersion ellipsoid appropriate to reproduce the observed fast rotators ($\gamma \approx 0$), and with the mass-follows-light hypothesis. They found an upper limit on $\beta $ by excluding values producing a negative $\vphim^2$ over a non-negligible region of the galaxy, clearly an unphysical solution.  Our analysis can provide an interpretation for the existence of this limit. In fact, from equation (\ref{eq:gamma_def}), the request of positivity for $\vphim^2$ with $\gamma\approx 0$ translates into the request of $\Deltac \gtrsim 0$, and then of $b \lesssim b_0$.  Interestingly, for the mass models considered in this work, $b_0$ increases with the flattening [see equations (\ref{eq:b0_MN}), (\ref{eq:b0_S}), (\ref{eq:binney_zc0}) and (\ref{eq:B_b0C_exp})], in agreement with the trend shown by the maximum values estimated for $\beta $ (of $\beta_{\rm max} \sim 0.7 \epsilon_{\mathrm{intr}})$.

\section*{Acknowledgements}

The referee is thanked for useful comments, that improved the clarity of the paper. We acknowledge financial support from the Project PRIN MUR 2022 (code 2022ARWP9C) \lq\lq Early Formation and Evolution of Bulge and HalO (EFEBHO)", PI: Marconi, M., funded by European Union – Next Generation EU (CUP J53D23001590006).

\section*{Data availability}

No datasets were generated or analyzed in support of this research.



\bibliographystyle{mnras}
\bibliography{dd24} 




\appendix

\section{The most general $\vphiqm $ decomposition, and an alternative ansatz for the radial JE}\label{app:gen_ans}

\subsection{The $\lambda$-decomposition for $\vphiqm$}\label{sec:lambda_dec}

Once the necessary condition $\vphiqm \ge 0$ is satisfied by the solution of equation (\ref{eq:jeans_rad_anis}), a choice is required to decompose $\vphiqm$ in its streaming $\overline{v_{\varphi}}^2$ and dispersion $\sigma_{\varphi}^2$ components. 
By construction, the Satoh decomposition presented in Sections 2.1 and 3.1  requires $\Delta$ (or $\Deltac$) $\ge 0$. Clearly, if $\Delta$ (or $\Deltac$) is negative over some region of the ($R,z$)-plane, a generalized Satoh decomposition can be applied just by imposing $k(R,z) = 0$ over this region. However, this decomposition is quite \lq\lq rigid", leading necessarily to $\overline{v_{\varphi}} = 0$ and $\sigma_{\varphi}^2 = \vphiqm$ there.
It is therefore useful to provide the most general decomposition of a positive $\vphiqm$: we refer to this as the $\lambda$-decomposition. It is simple to realize that \textit{all} the decompositions are of the form

 \begin{equation}
     \overline{v_{\varphi}} = \lambda \sqrt{\vphiqm} \; \; , \; \; \sigma_{\varphi}^2 = (1-\lambda^2)\vphiqm,
     \label{eq:l_dec}
 \end{equation}

\noindent
where in full generality $0 \le \lambda^2(R,z) \le 1$, with negative values of $\lambda$ corresponding to regions of \lq\lq counter-rotation". Two limit cases are possible: $\lambda = 0$ corresponds to no net rotation in the azimuthal direction, and $\lambda^2 = 1$ to a maximally rotating model. The value of $\lambda$ corresponding to the isotropic rotator is obtained by imposing $\sigma_{\varphi}^2 = \sigma_z^2$ from the second of equation (\ref{eq:l_dec}), so that

\begin{equation}
    \lambda_{\text{I}}^2(R,z) \equiv \frac{\Delta}{\vphiqm},
    \label{eq:l_iso}
\end{equation}

\noindent
where $\Delta = \vphiqm - \sigma_z^2 \ge 0$ for isotropy is given by the second of equation (\ref{eq:2int_jeans}).

Of course, being the $\lambda$-decomposition the most general possible, any other decomposition of $\vphiqm$ must be related to it. For example, if $\Delta \ge 0$ and the Satoh decomposition is hence possible, the following relation between $\lambda$ and $k$ holds:

 \begin{equation}
     \lambda = k \sqrt{\frac{\Delta}{\vphiqm}}.
     \label{eq:lambda_k_rel}
 \end{equation}

This shows how, except for the non-rotating case, a spatially constant $k$ does \textit{not} necessarily correspond to a spatially constant $\lambda$, and vice versa.

\subsection{The $\mu$-ansatz for $\Deltac$}\label{sec:mu_ans}

As discussed in Section \ref{sec:gen_sys}, an ansatz is required to solve the radial JE, and the consequences of assuming the $b$-ansatz are discussed in this work. We consider here a second family of ansatz (\lq\lq\textit{$\mu$-ansatz}") in which priority is given to the relation between $\Deltac$ and the unknown $\sigma_R^2$, so that:

\begin{equation}
    \Deltac = \mu \sigma_R^2,\quad \vphiqm = (1+\mu)\sigma_R^2,
    \label{eq:mu_ans_def1}
\end{equation}

\noindent
with the obvious constraint that $\mu(R,z) \ge -1$. Given this choice, it is possible to solve in closed form the radial JE and obtain $\sigma_R^2$. In fact, combining equations (\ref{eq:jeans_rad_anis}) and (\ref{eq:mu_ans_def1}) leads to a linear, non-homogeneous differential equation prone to an explicit solution, similar to that for anisotropic spherically symmetric systems (\citealt{Binney1982MNRAS.200..361B}, C21). Assuming a vanishing radial \lq\lq pressure" at infinity, the solution of the equation reads 

\begin{equation}
    \rhos\sigma_R^2 = \int_R^{\infty} \rhos(u,z)\frac{\partial\Phi(u,z)}{\partial u} e^{-g(u,R,z)} du,
    \label{eq:mu_ans_ode_sol}
\end{equation}

\noindent
where

\begin{equation}
    g(u,R,z) \equiv \int_{u}^R \frac{\mu(v,z)}{v} dv,
    \label{eq:mu_ans_g_def}
\end{equation}

\noindent
and so, if $\mu$ is independent of $R$,

\begin{equation}
    \rhos\sigma_R^2 = \frac{1}{R^{\mu}}\int_R^{\infty}\rho(u,z)\frac{\partial\Phi(u,z)}{\partial u} u^{\mu} du.
    \label{eq:mu_ans_ode_sol_part}
\end{equation}

Quite obviously, once the JEs are solved, each adopted ansatz can be expressed in terms of the others. For example, the $\mu$ function associated with a $b$-ansatz is given by

\begin{equation}
    \mu = \frac{\Deltac}{b\sigma_z^2},
    \label{eq:b_ans_mu}
\end{equation}

\noindent
further reducing to $\mu = \Delta/\sigma_z^2$ in the two-integral case, when $b=1$ and $\Deltac = \Delta$. Of course, if one adopts the $\mu$-ansatz, the equivalent $b$ is just given by equation (\ref{eq:b_ans_def}), where now $\sigma_R^2$ is computed as in equation (\ref{eq:mu_ans_ode_sol}).

An important advantage of using the $\mu$-ansatz should be noticed. From equation (\ref{eq:mu_ans_ode_sol}) it follows that $\sigma_R^2 \ge 0$ for realistic potentials with $\partial \Phi/\partial R \ge 0$. Therefore, the condition $\mu \ge 1$ in the second of equation (\ref{eq:mu_ans_def1}) automatically assures the positivity of $\sigma_R^2$ \textit{and} of $\vphiqm$, without the need for complicated analytical or numerical checks. Therefore, the combination of the $\mu$-ansatz with the generalized Satoh $k$-decomposition, or better with the most general $\lambda$-decomposition, allows for a relatively simple and flexible closure of the JEs.

\section{Relations between the regions $\mathscr{B}^{\pm}$, $\mathscr{C}^{\pm}$, $\mathscr{D}^{\pm}$, and between the limiting values for {\lowercase{\it b}} }\label{app:tech_res}

Some general results about the $\mathscr{B}$, $\mathscr{C}$ and $\mathscr{D}$ regions of axisymmetric systems described by the JEs in Section \ref{sec:gen_sys}, with $\partial \Phi/\partial R \ge 0$ everywhere, and restricting to $b = b(z)$, are proved.

\subsection{The ansatz-independent sets}
\label{app:sets}

The following ansatz-independent sets, based on the sign of the commutator, are important for the following discussion:

\begin{equation}
\begin{split}
    \omegapm &= \left\{ (R,z) : \left[ \rhos, \Phi \right] \gtrless 0 \right\},\\
    \omegazero &= \left\{ (R,z) : \left[ \rhos, \Phi \right] = 0 \right\};
\end{split}
\label{eq:omega_def}
\end{equation}

\noindent
we also define $\opmzero = \omegapm \cup \omegazero$.

\begin{theorem}
\label{thm:0}
For all systems, $\bminus \subseteq \dminus$, and $\omegaminus \subseteq \dminus$.

\end{theorem}
\begin{proof}
    Let $\mathrm{P} = (R,z) \in \bminus$. From the first of equation (\ref{eq:b_def}) and given $\sigma_z^2 \ge 0$, then $\mathrm{P} \in \dminus$, and the first inclusion is proved. Let $\mathrm{P} = (R,z) \in \omegaminus$. From equation (\ref{eq:jeans_identity}) and given $\partial\Phi/\partial R \ge 0$, then $\mathrm{P} \in \dminus$, and the second inclusion is proved. From the proved results immediately follows that $\dpzero \subseteq \bpzero$, and $\dpzero \subseteq \opzero$.
\end{proof}

\begin{theorem}
\label{thm:2}
If $\bminus \subseteq \opzero$, then $\bminus \subseteq \cplus$. If $\cminus \subseteq \opzero$, then $\cminus \subseteq \bplus$. 
\end{theorem}
\begin{proof}
    Let $\mathrm{P} = (R,z) \in \bminus$, and $\left[\rhos,\Phi\right] \ge 0$ there. From the second of equation (\ref{eq:delta_bans}), with $B < 0$ and $\left[\rhos,\Phi\right] \ge 0$, then $\mathrm{P} \in \cplus$, and the first part the Theorem follows immediately if $\left[\rhos,\Phi\right] \ge 0$ for all points in $\bminus$. The second part is proved similarly, by considering $C < 0$ in equation (\ref{eq:delta_bans}).
\end{proof}

\begin{theorem}
\label{thm:3}
    If $\bpzero \subseteq \omegaminus$, then $\bpzero \subset \cminus$. If $\cpzero \subseteq \omegaminus$, then $\cpzero \subset \bminus$.
\end{theorem}
\begin{proof}
    Let $\mathrm{P} = (R,z) \in \bplus$, and $\left[\rhos,\Phi\right] < 0$ there. From the second of equation (\ref{eq:delta_bans}), with $B \ge 0$ and $\left[\rhos,\Phi\right] < 0$, then $\mathrm{P} \in \cminus$, and the first part the Theorem follows immediately if $\left[\rhos,\Phi\right] < 0$ for all points in $\bpzero$. The second part is proved similarly, by considering $C \ge 0$ in equation (\ref{eq:delta_bans}). 
\end{proof}

\subsection{Relations between the limits for $b$}\label{app:bz_thm}

\begin{theorem}
\label{thm:4}
For all systems, $b_0(z) \le \bM(z)$ over the rectangular strip $R \ge 0$ and $z \in \projz(\bminus)$. 
\end{theorem}

\begin{proof}
    From equation (\ref{eq:bmz_criterion}), $\bM(z)$ is defined over the strip in the ($R,z$)-plane with $z \in \projz(\bminus)$, and from equation (\ref{eq:b0z_criterion}) $b_0(z)$ is defined over the strip with $z \in \projz(\dminus)$. From Theorem \ref{thm:0}, $\bminus \subseteq \dminus$, so that $b_0(z)$ is certainly defined for $z \in \projz(\bminus)$. Over $\bminus$, both $B$ and $D$ are negative, so that from the first of equation (\ref{eq:b_def}), $\lvert B \rvert = R \lvert D \rvert/\rhos - \sigma_z^2 \le R\lvert D \rvert/\rhos$. It follows that the minimum in the first of equation (\ref{eq:bmz_criterion}) over $\bzminus$ is larger than the minimum obtained when $\lvert B \rvert$ is replaced by $R\lvert D \rvert/\rhos$, and in turn this minimum is larger than that evaluated over the wider set $\dzminus \supseteq \bzminus$: from the first of equation (\ref{eq:b0z_criterion}), this last minimum is the value of $b_0(z)$ for points with $z \in \projz(\bminus)$, concluding the proof. 
\end{proof}

\begin{theorem}
\label{thm:b1}
Let $\cminus \subseteq \opzero$: then $b_1(z) \le 1$ over the rectangular strip $R \ge 0$ and $z \in \projz(\cminus)$. 
\end{theorem}
\begin{proof}
    For assumption $\cminus \subseteq \opzero$, so that from Theorem \ref{thm:2} $\bplus \cap \cminus = \cminus$, and from equation (\ref{eq:b1_def}) with $\left[ \rhos, \Phi \right] \ge 0$, $b_1(z) \le 1$ for $z \in \projz(\cminus)$. 
\end{proof}

\begin{theorem}
\label{thm:bMb0b2}
Let $\bminus \subseteq \opzero$: then $b_2(z) \ge 1$ and $b_0(z) \le b_2(z) \le \bM(z)$ over the rectangular strip $R \ge 0$ and $z \in \projz(\bminus)$.

\end{theorem}
\begin{proof}
    From equation (\ref{eq:b2_def}) and Theorem \ref{thm:2}, $b_2(z)$ is defined over $\bminus \cap \cplus = \bminus$, and from the assumed positivity of $\left[ \rhos, \Phi \right]$ over $\bminus$, it follows that $b_2(z) \ge 1$, for $z \in \projz(\bminus)$. Moreover, from Theorem \ref{thm:4} the functions $b_2(z)$, $b_0(z)$ and $\bM(z)$ are all defined for points with $z \in \projz(\bminus)$. Again from equation (\ref{eq:b2_def}), using the inequality $\lvert B \rvert \le R\lvert D \rvert/\rhos$ (see the proof of Theorem \ref{thm:4}), and finally from the positivity of the commutator over $\bminus$, it follows that $b_2(z)$ is larger than the minimum of the function appearing in the second of equation (\ref{eq:b0z_criterion}), now computed over $\bminus$. As $\bminus \subseteq \dminus$ from Theorem \ref{thm:0}, extending this minimum over $\dminus$ further decreases its value, proving that $b_2(z) \ge b_0(z)$. $\bM(z) \ge b_2(z)$ follows by comparison between equations (\ref{eq:bmz_criterion}) and (\ref{eq:b2_def}), where $\bminus \cap \cplus = \bminus$ and $\sigma_z^2 \ge 0$. Notice that $b_0(z)$ can be larger or lower than $1$ depending on the sign of $\left[ \rhos, \Phi \right]$ over the region $\dminus \cap \bplus$, for $z \in \projz(\bminus)$.
\end{proof}

\begin{theorem}
\label{thm:b0}
   Let $\dminus \subseteq \opzero$: then $b_0(z) \ge 1$ over the rectangular strip $R \ge 0$ and $z \in \projz(\dminus)$, and $1 \le b_0(z) \le b_2(z) \le \bM(z)$ over the strip $z \in \projz(\bminus) \subseteq \projz(\dminus)$.
\end{theorem}

\begin{proof}
    By assumption $\dminus \subseteq \opzero$, so that from equation (\ref{eq:b0z_criterion}) $b_0(z) \ge 1$ for $z \in \projz(\dminus)$. The second part of the Theorem descends immediately from the inequality just proved and Theorem \ref{thm:bMb0b2}, because from Theorem \ref{thm:0} $\bminus \subseteq \dminus$ ($\subseteq \opzero$ by assumption). Notice how strengthening the hypothesis of Theorem \ref{thm:bMb0b2} (i.e. $\bminus \subseteq \opzero$) to that of the present case ($\dminus \subseteq \opzero$) allows to establish whether $b_0(z)$ is larger or smaller than unity, then resolving the indeterminacy mentioned at the end of the proof of Theorem \ref{thm:bMb0b2}.
\end{proof}


\bsp	
\label{lastpage}
\end{document}